\documentclass{article}
\usepackage{amsthm, amssymb, amsmath}
\usepackage{enumerate}
\usepackage{amsthm}
\usepackage{graphicx}
\usepackage{xcolor}
\numberwithin{equation}{section}
\newtheorem{theorem}{Theorem}[section]

\newtheorem{lemma}[theorem]{Lemma}

\newtheorem{proposition}[theorem]{Proposition}
\newtheorem{corollary}[theorem]{Corollary}

\huge

\begin{document}
\title{Exact characterizations for quantum conditional mutual information and some other entropies}
\author{Zhou Gang\footnote{gangzhou@binghamton.edu, partly supported by Simons travel support 709542.}}
\maketitle
\centerline{Department of Mathematics and Statistics, Binghamton University, Binghamton, NY, 13902
}
\setlength{\leftmargin}{.1in}
\setlength{\rightmargin}{.1in}
\normalsize \vskip.1in
\setcounter{page}{1} \setlength{\leftmargin}{.1in}
\setlength{\rightmargin}{.1in}
\large

\date

\setlength{\leftmargin}{.1in}
\setlength{\rightmargin}{.1in}
\normalsize \vskip.1in
\setcounter{page}{1} \setlength{\leftmargin}{.1in}
\setlength{\rightmargin}{.1in}
\large

\section*{Abstract}
Lieb and Ruskai's strong subadditivity theorem, which shows that the conditional mutual information must be nonnegative, is fundamental in quantum theory. It has numerous applications, such as in quantum error correction. When the mutual information is zero, the Petz recovery map can be used to reconstruct the quantum channel. When the mutual information is small, one seeks to define an optimal recovery channel. To this end, a mathematical characterization of the mutual information is desirable. We address this problem by providing an exact characterization of the mutual information, along with characterizations for other entropies. Our controls are sharp, leaving no room for improvement, in the sense that we provide equalities, regardless of whether the mutual information (or remainder) is small or large.
We transform the definitions of these entropies into a summation of explicitly constructed terms, and the definition of each term obviously demonstrates the desired positivity/convexity/concavity. The summation converges rapidly and absolutely in a chosen elementary norm.
 \vskip 4cm

\tableofcontents

\section{Introduction}
In this paper, we investigate the following two related problems.

The first problem is the joint concavity of the mapping 
\begin{align}\label{eq:jointMap}
  (A, B)\rightarrow\ \text{Tr} (K^* A^{q} K B^r)
\end{align} where $A$ and $B$ are positive definite $N\times N$ and $M\times M$ matrices, and $K$ is an $N \times M$ matrix, $q,$ $r$ are nonnegative scalar with $q + r \le 1$.

The second problem is the quantum conditional mutual information
\begin{equation}\label{def:mutual}
    I(A:C|B)_\rho := S(\rho_{AB}) + S(\rho_{BC}) - S(\rho_{ABC}) - S(\rho_B)
\end{equation}
where $\rho_{ABC}$ is a full ranked density matrix on the system $\mathcal{H}_A \otimes \mathcal{H}_B \otimes \mathcal{H}_C$; for any density matrix $\rho$
\begin{equation}
    S(\rho) := -\text{Tr}(\rho \log_2 \rho)
\end{equation}
$\rho_{AB} := \text{Tr}_C(\rho_{ABC})$, and $\rho_{BC}$ and $\rho_B$ are defined similarly.

In \cite{LIEB1973267}, Lieb proved that the mapping \eqref{eq:jointMap} is jointly concave, a fundamental result known as the Concavity Theorem. In \cite{LiebRuskai1973}, Lieb and Ruskai proved that $I(A: C|B)_\rho \ge 0$. These results are fundamentally important due to their numerous applications, as seen in \cite{ohya2004quantum, MRenes2022}, and were described as ``the key result on which virtually every nontrivial quantum coding theorem relies," as noted in \cite{LindenWinter2005}. There are many different proofs; see \cite{Ruskai2022,Ruskai2007,FawziSaunderson2017,simon2005trace, nielsen2004simple, Effros2009,Aujla2011, BeaudryRenner2012,Ismail2013}. The techniques were used in different places, see e.g.\cite{FrankLieb2013, FrankLieb2013b, CarlenFrankLieb2018, Hiai2013, Hiai2016}.

When $I(A:C|B)_\rho = 0$, Petz proved in \cite{Petz1986, Petz1988} that the quantum channel can be reconstructed based on $\rho_{ABC}$. This map, called the Petz recovery map, plays a central role in quantum error correction, as seen in the syndrome-based recovery map, for example \cite{BiswasEtc2024, biswas2025universal}.

We are interested in the generic case $I(A: C|B)_\rho \gneq 0$. Here, one would like to find an ``optimal" reconstruction for the quantum channel, or tailored reconstructions if specific information about the density matrix $\rho_{ABC}$ is available. This turns out to be difficult; a major obstacle is the lack of a characterization for the mutual information. As pointed out in \cite{FawziRenner2015}, see also \cite{Sutter2018}, ``a natural question that is very relevant for applications is to characterize states for which the conditional mutual information is approximately zero". There is a rich literature on related recovery channels in quantum error correction, besides quantum computations as mentioned above, for topological entanglement entropy, see \cite{KitaevPreskill2006, LevinWen2006, ShiKatoKim2020, MLevin2024, KimLevinEtc2023, KatoBrandao2020} and for recovery maps and/or reversibilities in quantum communications; see \cite{LindenWinter2005, Brandao2010, FawziRenner2015, Sutter2018, SutterFawziRenner2016, CarlenVershynina2018, DattaWilde2015, JuneRenner2018, HaydenJosPetzWinter2004, berta2014monotonicity, Wilde2015}.

In this paper, we address the problem of finding an exact characterization for quantum conditional mutual states and Lieb's concavity theorem, as seen in \eqref{eq:mutual} within Theorem \ref{MainTHM:Mutual} and Theorem \ref{MainTheorem:concavity}. Unlike remainder estimates, as in \cite{CarlenLieb2012, CarlenLieb2014, FawziRenner2015, SutterFawziRenner2016, JuneRenner2018, CarlenVershynina2020}, we provide explicit constructions and equalities.

Technically, we start with an improved understanding of geometric means, detailed in Theorem \ref{THM:secondOrder}. We derive an exact formula, and its definition clearly demonstrates the desired property. Subsequently, we utilize this, along with other techniques and Ando's approach from \cite{ANDO1979203}, to find exact characterizations for Lieb's concavity theorem and Lieb and Ruskai's strong subadditivity theorem.

In future papers, we will address the ``optimal" recovery channel.

For illustration, we describe our results for the geometric means of two positive-definite matrices. Let $H(\epsilon):=M_{0}(A+\epsilon X, B+\epsilon Z)$ be the geometric mean of positive definite matrices $A+\epsilon X$ and $B+\epsilon Z$, i.e., $H(\epsilon)$ is the maximal Hermitian matrix such that
\begin{equation*}
\Pi(\epsilon):=
\begin{pmatrix}
A+\epsilon X & H(\epsilon) \\
H(\epsilon) & B+ \epsilon Z
\end{pmatrix}\geq 0.
\end{equation*} 
Here, $A$ and $B$ are positive definite, $X$ and $Z$ are Hermitian, and $\epsilon\in \mathbb{R}$ satisfies $|\epsilon|\ll 1$. The existence of the "maximal" $H(\epsilon)$ is well known, see e.g., \cite{carlen2010trace}. In this paper, we use the convention that $A\geq B$ signifies that $A-B$ is positive semi-definite.

It is not hard to see that 
\begin{equation}\label{eq:MeanConcave}
 H(0)\geq \frac{1}{2}(H(\epsilon)+H(-\epsilon))
\end{equation}
because $H(0)$ is the maximal one to make $\Pi(0)>0$ and
\begin{align*}
\begin{pmatrix}
A & \frac{1}{2} \big(H(\epsilon)+H(-\epsilon)\big) \\
\frac{1}{2} \big(H(\epsilon)+H(-\epsilon)\big) & B
\end{pmatrix}
=\frac{1}{2}\big(\Pi(\epsilon)+\Pi(-\epsilon)\big)\geq 0.
\end{align*}

In this paper, specifically, in Theorem \ref{THM:secondOrder} below, we prove constructively: $$\frac{d^2}{d\epsilon^2} H(\epsilon)|_{\epsilon=0}=-\text{Cross}_{A,B}(X,Z)\leq 0.$$ More importantly, instead of relying on arguments like \eqref{eq:MeanConcave}, our definition of $\text{Cross}_{A,B}(X,Z)$ in \eqref{def:Cross} shows that it is obviously positive semi-definite.

The paper is organized as follows. In Section \ref{sec:tool}, we study the geometric means of matrices. In Sections \ref{sec:convexity}, \ref{sec:relative}, and \ref{sec:Mutual}, we use this tool and some other techniques to study Lieb's concavity Theorem, joint convexity of relative entropy, and strong additivity of quantum entropy.

In this paper, we use the following conventions: If $A$ and $B$ are Hermitian matrices, then $$A>B\ (A\geq B)$$ signifies that $A-B$ is positive definite (positive semi-definite). We use the notation $$C=C(\Psi,\Phi)$$ to signify that $C$ is a constant whose size only depends on $\Psi$ and $\Phi$.

We adopt the following norm for any $N\times N$ matrix $K$:
\begin{align}\label{def:Norm}
\begin{split}
  \|K\|:=&\sup\big\{ |\langle K x, Kx\rangle|^{\frac{1}{2}} \ \big|\ x\in \mathbb{C}^N,\ \|x\|=1\big\}\\
  =&\max\big\{ \lambda\geq 0\ \big| \ \lambda^2\ \text{is an eigenvalue of }\ K^* K\big\}.
\end{split}
\end{align} And hence, a Banach space is defined accordingly.

    
\section{Geometric mean}\label{sec:tool}
In this section, we study the geometric mean and present one of the main technical tools of the present paper.

Given positive definite $N\times N$ matrices $A$ and $B,$ $N\times N$ Hermitian matrices $X$ and $Z$, and a real scalar $\epsilon$ with $|\epsilon|\ll 1$, we are interested in studying the geometric mean of $A+\epsilon X$ and $B+\epsilon Z$, i.e. studying the $N\times N$ matrix $H$:
\begin{equation}\label{def:GeoMean}
  H(\epsilon):=M_{0}(A+\epsilon X, B+ \epsilon Z)
\end{equation} which is defined as the maximal positive semi-definite matrix to make the
$2N\times 2N$ block matrix positive semi-definite:
\begin{equation}\label{def:Pi}
\Pi:=\begin{pmatrix}
A+\epsilon X & H(\epsilon) \\
H(\epsilon) & B+ \epsilon Z
\end{pmatrix}\geq 0.
\end{equation}

As discussed earlier
$
  \frac{d^2}{d\epsilon^2}H(\epsilon)\Big|_{\epsilon=0}\leq 0$ because of \eqref{eq:MeanConcave}. Here we want to derive an explicit expression demonstrating this in an obvious way, see Theorem \ref{THM:secondOrder} below.

Before stating the result, we define a few Hermitian matrices
\begin{align}\label{def:PsiE1E2}
\begin{split}
\Psi :=& \left( \frac{A+B}{2} \right)^{-\frac{1}{2}} \frac{A-B}{2} \left( \frac{A+B}{2} \right)^{-\frac{1}{2}},\\
E_2 :=& \left( \frac{A+B}{2} \right)^{-\frac{1}{2}} \frac{ X-  Z}{2} \left( \frac{A+B}{2} \right)^{-\frac{1}{2}},\\
E_1:=&-\left( \frac{A+B}{2} \right)^{-\frac{1}{2}} \frac{ X+  Z}{2} \left( \frac{A+B}{2} \right)^{-\frac{1}{2}}.
\end{split}
\end{align} Obvious, since $A$ and $B$ are positive definite,
\begin{align}\label{eq:shrink}
  -I<\Psi<I.
\end{align}
Based on \eqref{def:PsiE1E2}, we define the following Hermitian matrices
\begin{align}\label{def:PsiOmega}
\begin{split}
\widetilde{\Psi}(s):=&\sum_{k = 0}^{+\infty} 
(1+s)^{-k} \Psi^{\,2k}\geq 0,\\    
\Omega(s) := &
(1+s)^{-\frac{1}{2}}\, (E_2 + E_1 \Psi)
\;+\;
(1+s)^{-1 -\frac{1}{2}}
\Psi (E_2 + E_1 \Psi)\,\Psi .
\end{split}
\end{align} By \eqref{eq:shrink}, the summation in the definition of $\widetilde{\Psi}$ converges rapidly, even at $s=0.$ 

Lastly, define a matrix which is obviously positive semi-definite and will play a central role in the present paper: recall that the matrices $A$, $B$, $\widetilde{\Psi}$ and $\Omega$ are Hermitian,
\begin{align}\label{def:Cross}
\begin{split}
& \text{Cross}_{A,B}( X, Z) \\
:=&\frac{1}{\pi}(\frac{A+B}{2})^{\frac{1}{2}}\Bigg[\int_{0}^{\infty} s^{-\frac{1}{2}}(1+s)^{-1}\widetilde{\Psi}(s)\Omega(s) \widetilde{\Psi}(s) \Omega^{*}(s) \widetilde{\Psi}(s)\ ds\\
&+\Psi\Big(\int_{0}^{\infty} s^{-\frac{1}{2}}(1+s)^{-3}\widetilde{\Psi}(s)(E_2+E_1\Psi)\ (E_2+E_1\Psi)^{*}\  \widetilde{\Psi}(s)\ ds \Big)\Psi\Bigg](\frac{A+B}{2})^{\frac{1}{2}}.
\end{split}
\end{align}

We want to emphasize that, even though the definition \eqref{def:Cross} involves infinite $s-$integrals, the integrals are easy to evaluate because every integrand is a summation of terms with the form $s^{-\frac{1}{2}}(1+s)^{-k} F_k $ where $k\geq 1$ and $F_k$ is independent of $s$. This expression in \eqref{def:Cross} is preferred because it demonstrates that, obviously, $\text{Cross}_{A, B}(X, Z)$ is positive semi-definite.
 
The main result of this section is the following:
\begin{theorem}\label{THM:secondOrder}
\begin{equation}\label{eq:Hvare}
\frac{1}{2}\frac{d^2}{d\epsilon^2}H(\epsilon)\Big|_{\epsilon=0}=-\text{Cross}_{A,B}(X,Z)\leq 0.
\end{equation}
And
\begin{equation}\label{eq:iff}
  \text{Cross}_{A,B}(X,Z)= 0 \ \text{if and only if}\ 
E_2 + E_1 \Psi=0.
\end{equation}
\end{theorem}
This will be proved in the rest of the section.

Next, we want to prove that the geometric mean is jointly concave. Define 
\begin{align}
  H(t):=M_0(tA_1+(1-t)A_2, \ tB_1+(1-t)B_2),
\end{align} where $A_1$, $A_2$, $B_1$, $B_2$ are positive definite $N\times N$ matrices.
Apply Theorem \ref{THM:secondOrder} to obtain the following result:
\begin{corollary}
\begin{equation}
  \frac{1}{2} H^{''}(t)=\frac{1}{2}\frac{d^2}{d\epsilon^2}H(t+\epsilon)\Big|_{\epsilon=0}=-\text{Cross}_{tA_1+(1-t)A_2,\ tB_1+(1-t)B_2}(A_1-A_2, B_1-B_2)\leq 0.
\end{equation}
\end{corollary}


\subsection{Reformulation of Theorem \ref{THM:secondOrder}}
By some straightforward calculation, see e.g. \cite{carlen2010trace}, for any positive definite matrices $V$ and $W$, their geometric means is
\begin{equation}\label{def:formuMeans}
M_0 (V, W) = W^{\frac{1}{2}} (W^{-\frac{1}{2}} V W^{-\frac{1}{2}})^{\frac{1}{2}} W^{\frac{1}{2}}.
\end{equation}

However this formula is hardly useful for our specific problem, because one can not tell that $H^{''}(0)\leq 0$. 

Thus, we have to find a new approach.

We start with transforming the matrix $\Pi$, defined in \eqref{def:Pi}, into a convenient form:
$$
\Pi = \begin{bmatrix} 
\Gamma + \Theta & H(\epsilon)\\ 
H(\epsilon) & \Gamma - \Theta 
\end{bmatrix}
$$ with $\Gamma$ and $\Theta$ defined as
$$\Gamma := \frac{A+ B+ \epsilon (X+Z)}{2}$$ and $$\Theta := \frac{A - B+ \epsilon (X-Z)}{2}.$$

Multiple both sides of the matrix $\Pi$ by $\begin{bmatrix}
\Gamma^{-\frac{1}{2}} & 0 \\
0 & \Gamma^{-\frac{1}{2}}
\end{bmatrix}$ and compute directly to find
\begin{align}
\begin{bmatrix}
\Gamma^{-\frac{1}{2}} & 0 \\
0 & \Gamma^{-\frac{1}{2}}
\end{bmatrix}
\Pi
\begin{bmatrix}
\Gamma^{-\frac{1}{2}} & 0 \\
0 & \Gamma^{-\frac{1}{2}}
\end{bmatrix}
&=
\begin{bmatrix} 
I + \Gamma^{-\frac{1}{2}} \Theta \Gamma^{-\frac{1}{2}} & \Gamma^{-\frac{1}{2}} H(\epsilon) \Gamma^{-\frac{1}{2}} \\ 
\Gamma^{-\frac{1}{2}} H(\epsilon) \Gamma^{-\frac{1}{2}} & I - \Gamma^{-\frac{1}{2}} \Theta \Gamma^{-\frac{1}{2}} 
\end{bmatrix}.
\end{align}

This transformation produces a significant advantage: the two matrices $I + \Gamma^{-\frac{1}{2}} \Theta \Gamma^{-\frac{1}{2}}$ and $I - \Gamma^{-\frac{1}{2}} \Theta \Gamma^{-\frac{1}{2}} $ commute!

Before taking square roots on $I + \Gamma^{-\frac{1}{2}} \Theta \Gamma^{-\frac{1}{2}} $ and $I - \Gamma^{-\frac{1}{2}} \Theta \Gamma^{-\frac{1}{2}} $, we need to prove they are positive definite. Indeed, since $A+\epsilon Y$ and $B+\epsilon Z$ are positive definite, we have $-\Gamma<\Theta<\Gamma$, and hence 
\begin{align}\label{eqn:compareI}
    -I<\Gamma^{-\frac{1}{2}} \Theta \Gamma^{-\frac{1}{2}}<I.
\end{align} 

Continue to analyzing $H(\epsilon).$
If $H(\epsilon)  $ makes $\Pi\geq 0$, then we must have
\begin{align}\label{eq:GammaHGamma}
\begin{split}
\Gamma^{-\frac{1}{2}} H(\epsilon)  \Gamma^{-\frac{1}{2}} \leq \sqrt{I+\Gamma^{-\frac{1}{2}} \Theta \Gamma^{-\frac{1}{2}}} \sqrt{I-\Gamma^{-\frac{1}{2}} \Theta \Gamma^{-\frac{1}{2}}} &= \sqrt{I - \Gamma^{-\frac{1}{2}} \Theta \Gamma^{-1} \Theta \Gamma^{-\frac{1}{2}}} \\
&= \sqrt{I - \Lambda}
\end{split}
\end{align}
where $\Lambda$ is a positive semi-definite matrix defined as 
\begin{align}\label{eq:sizeZ}
    0\leq \Lambda:=\Gamma^{-\frac{1}{2}} \Theta \Gamma^{-1} \Theta \Gamma^{-\frac{1}{2}}=(\Gamma^{-\frac{1}{2}} \Theta \Gamma^{-\frac{1}{2}})^2<I,
\end{align} and the bound $ \Lambda< I$ is implied by \eqref{eqn:compareI}.
Since $H(\epsilon)$ is the maximal Hermitian matrix to satisfy \eqref{eq:GammaHGamma} and $\Gamma$ is positive definite,
\begin{align}
   H(\epsilon)   = \Gamma^{\frac{1}{2}} \sqrt{I-\Lambda} \Gamma^{\frac{1}{2}}.
\end{align}

It is hard to see directly that the $\epsilon^2$-term in $\sqrt{I-\Lambda}$ has a favorable sign. For this reason, we use an integral representation, see Lemma \ref{LM:integral} below, to derive a convenient form, 
\begin{align*}
\sqrt{I-\Lambda}
=& \frac{1}{\pi} \int_{0}^{\infty} s^{\frac{1}{2}} \left( s^{-1}I - \big(( 1+s)I-\Lambda\big)^{-1} \right)\, ds\\\\
=&\frac{1}{\pi} \int_{0}^{\infty} s^{\frac{1}{2}} \left( s^{-1}I - (1+s)^{-1}\big(I-(1+s)^{-1}\Lambda\big)^{-1} \right)\, ds.
\end{align*}
Thus,
\begin{equation}\label{eq:Dvare}
H( \epsilon)=\Gamma^{\frac{1}{2}} \sqrt{I-\Lambda} \Gamma^{\frac{1}{2}}=\frac{1}{\pi}\int_{0}^{\infty} s^{\frac{1}{2}} \left( s^{-1}\Gamma - (1+s)^{-1} \tilde{H}(\epsilon) \right)\, ds,
\end{equation} where $\tilde{H}(\epsilon)$ is defined as
\begin{equation}\label{def:tildeD}
\tilde{H}(\epsilon):=
\Gamma^{\frac{1}{2}}\big(I-(1+s)^{-1}\Lambda\big)^{-1}\Gamma^{\frac{1}{2}}.
\end{equation}
Observe that $\Gamma$ doesn't contain any term of order $O(\epsilon^2)$, hence is of no interest to us.

In what follows we focus on $\tilde{H}(\epsilon)$.

Since $0\leq (1+s)^{-1}\Lambda<1$ by \eqref{eq:sizeZ}, we can expand in $(1+s)^{-1}\Lambda$, to find
\begin{align}\label{eq:expZ}
    \tilde{H}(\epsilon)=\sum_{l=0}^{\infty} (1+s)^{-l} \Gamma^{\frac{1}{2}}\Lambda^l \Gamma^{\frac{1}{2}}=\sum_{l=0}^{\infty} (1+s)^{-l} (\Theta\Gamma^{-1})^{2l-1}\Theta.
\end{align} 

Transform the expression on the right hand side into a convenient form. Observe
\begin{align*}
\Gamma^{-1}  = \left( \frac{A + B + \epsilon X + \epsilon Z}{2} \right)^{-1}
= \left( \frac{A+B}{2} \right)^{-\frac{1}{2}} (I - \epsilon E_1)^{-1} \left( \frac{A+B}{2} \right)^{-\frac{1}{2}} \\
\end{align*}
where $E_1:=-\left( \frac{A+B}{2} \right)^{-\frac{1}{2}} \frac{ X+  Z}{2} \left( \frac{A+B}{2} \right)^{-\frac{1}{2}}$, as defined in \eqref{def:PsiE1E2}. And thus,
\begin{equation*}
\begin{aligned}
 (\Theta\Gamma^{-1})^{2l-1}\Theta = \left( \Theta \left( \frac{A+B}{2} \right)^{-\frac{1}{2}} (I - \epsilon E_{1})^{-1} \left( \frac{A+B}{2} \right)^{-\frac{1}{2}} \right)^{2l-1}\Theta.
\end{aligned}
\end{equation*}
We multiply to the left and right by $\left( \frac{A+B}{2} \right)^{-\frac{1}{2}}$ to find a convenient form
\begin{align}\label{def:UpsilonL}
\begin{split}
\Upsilon_l:=&\left( \frac{A+B}{2} \right)^{-\frac{1}{2}}  (\Theta\Gamma^{-1})^{2l-1}\Theta
\left( \frac{A+B}{2} \right)^{-\frac{1}{2}}\\
= &\left( (\Psi + \epsilon E_{2}) (I - \epsilon E_{1})^{-1} \right)^{2l-1} (\Psi + \epsilon E_{2})\\
=&\left( (\Psi + \epsilon E_{2}) \sum_{k=0}^{\infty} \epsilon^k E_1^{k} \right)^{2l-1} (\Psi + \epsilon E_{2}).
\end{split}
\end{align}
where $\Psi=\left( \frac{A+B}{2} \right)^{-\frac{1}{2}} \frac{A-B}{2} \left( \frac{A+B}{2} \right)^{-\frac{1}{2}}$, $E_{2}=\left( \frac{A+B}{2} \right)^{-\frac{1}{2}} \frac{ X-  Z}{2} \left( \frac{A+B}{2} \right)^{-\frac{1}{2}}$, as defined in \eqref{def:PsiE1E2}, is produced by $\left( \frac{A+B}{2} \right)^{-\frac{1}{2}} \Theta \left( \frac{A+B}{2} \right)^{-\frac{1}{2}}$
\begin{align*}
\Psi + \epsilon E_{2} =&\left( \frac{A+B}{2} \right)^{-\frac{1}{2}} \Theta \left( \frac{A+B}{2} \right)^{-\frac{1}{2}}\\
=& \left( \frac{A+B}{2} \right)^{-\frac{1}{2}} \left( \frac{A-B}{2} + \epsilon\frac{ X-  Z}{2} \right) \left( \frac{A+B}{2} \right)^{-\frac{1}{2}},
\end{align*} and in the last step we used that,
\begin{equation*}
    (I-\epsilon E_1)^{-1}=\sum_{k=0}^{\infty}\epsilon^k E_1^{k}.
\end{equation*}

$\Upsilon_0$ does not contain any $\epsilon^2-$term, and thus is of no interest to us. We only need to consider the cases $l\geq 1$. To illustrate, we compute $\Upsilon_1$ and $\Upsilon_2$ in detail, after that we develop a general theory.

For $\Upsilon_1$,
\begin{align*}
\Upsilon_1
  = 
    (\Psi + \epsilon E_2)\, (\sum_{k = 0}^{\infty}\epsilon^{k}E_1^k)\, (\Psi + \epsilon E_2).
\end{align*}
And thus, its second order term is $\epsilon^2 \bigl(E_2 + \Psi E_1)(E_2 + \Psi E_1)^* \geq 0.$ Here we used that, since $E_1,$ $E_2$ and $\Psi$ are Hermitian, 
\begin{equation}\label{eq:ThreeHerm}
E_2+E_1\Psi=(E_2+\Psi E_1)^*.
\end{equation}

For $\Upsilon_2,$
\begin{equation*}
\Upsilon_2 =   (\Psi + \epsilon E_2)
    \left( \sum_{l=0}^{\infty} \epsilon^l E_1^l \right)
    (\Psi + \epsilon E_2)
    \left( \sum_{k=0}^{\infty} \epsilon^k E_1^k \right)
    (\Psi + \epsilon E_2)
    \left( \sum_{j=0}^{\infty} \epsilon^j E_1^j \right)
    (\Psi + \epsilon E_2).
\end{equation*}
The coefficient of $\epsilon^2$ is a sum of the following ones: 
$$
 (E_2 + \Psi E_1)\, 
\Psi^2\, (E_2 + E_1 \Psi);\ 
\Psi (E_2 + \Psi E_1)\,
\Psi \, (E_2 + E_1 \Psi);
$$
$$
(E_2 + \Psi E_1)\, \Psi (E_2 + E_1 \Psi)\, \Psi;
\ 
\Psi^2 (E_2 + \Psi E_1)\, (E_2 + E_1 \Psi);
$$
$$
 (E_2 + \Psi E_1)\, (E_2 + E_1 \Psi)\Psi^2;\ 
\Psi (E_2 + \Psi E_1)\, (E_2 + E_1 \Psi)\Psi.
$$

We make preparation before stating the result for any $\Upsilon_l$, $l\in \mathbb{N}.$

Define, for any $l\in \mathbb{N}$,  
\begin{equation}\label{def:OmegaL}
\Xi_l :=  \sum_{j_1+j_2+j_3=2l-2} \Psi^{j_1} (E_2 + \Psi E_1) \Psi^{j_2} (E_2 + E_1 \Psi) \Psi^{j_3}.
\end{equation}
and
\begin{align}\label{def:Omegal12}
\begin{split}
\Xi_{l,1}:=&\sum_{ k_1 + k_2 = l-2} \Psi^{2 k_1 + 1} (E_2 + \Psi E_1) (E_2+E_1 \Psi ) \Psi^{2 k_2 + 1},\\\\
\Xi_{l,2}:= &\sum_{k_1 + k_2 + k_3 = l-1} \left[ \Psi^{2k_1} (E_2 + \Psi E_1) \right] \Psi^{2k_2} \left[ (E_2 + E_1 \Psi) \Psi^{2k_3} \right] \\
    &+ \sum_{k_1 + k_2 + k_3 = l-2} \left[ \Psi^{2k_1 + 1} (E_2 + \Psi E_1) \Psi \right] \Psi^{2k_2} \left[ (E_2 + E_1 \Psi) \Psi^{2k_3} \right] \\
    &+ \sum_{k_1 + k_2 + k_3 = l-2} \left[ \Psi^{2k_1 } (E_2 + \Psi E_1)  \right] \Psi^{2k_2} \left[\Psi (E_2 + E_1 \Psi) \Psi^{2k_3+1} \right] \\
    &+ \sum_{k_1 + k_2 + k_3 = l-3} \left[ \Psi^{2k_1 + 1} (E_2 + \Psi E_1) \Psi \right] \Psi^{2k_2} \left[ \Psi (E_2 + E_1 \Psi) \Psi^{2k_3 + 1} \right].
  \end{split}
\end{align}
Here, $k_1$, $k_2$, and $k_3$ are nonnegative integers.

The result for $\Upsilon_l$ is the following:
\begin{lemma}
The second order term in $\Upsilon_{l}$,
is $\epsilon^2\Xi_l$, with $\Xi_l$ defined above.

Because, in \eqref{def:OmegaL}, $j_1+j_2+j_3$ must be even, 
$\Xi_l$ can be decomposed into two parts:
\begin{equation}\label{eq:OmegaL}
\Xi_{l} = \Xi_{l,1}+\Xi_{l,2}.
\end{equation}
\end{lemma}
\begin{proof}
    We prove the result by a standard induction on $l$, with the base cases being $l=1,2$ considered above. It is straightforward, and we choose to skip it.
\end{proof}

Consequently, sum up all the $\Xi_l$ and use $(E_2 + \Psi E_1)^*=E_2+E_1\Psi$ to find
\begin{equation}
    \sum_{l \geq 1} \Xi_{l} (1+s)^{-l}=\sum_{l \geq 1} \Xi_{l, 1} (1+s)^{-l}+\sum_{l \geq 1} \Xi_{l, 2} (1+s)^{-l}
\end{equation}
and the two terms on the right hand side take simple forms:
\begin{equation*}
\sum_{l \geq 1} \Xi_{l, 1} (1+s)^{-l} = (1+s)^{-2} \Psi \tilde{\Psi}(s) (E_2 + \Psi E_1) (E_2 + \Psi E_1)^* \tilde{\Psi}(s) \Psi,
\end{equation*}
and
\begin{equation*}
\sum_{l \geq 1} \Xi_{l, 2} (1+s)^{-l} = \tilde{\Psi}(s) \Omega(s) \tilde{\Psi}(s) \Omega^*(s) \tilde{\Psi}(s)
\end{equation*}
where $\tilde{\Psi}(s)$ and $\Omega(s)$ are defined in \eqref{def:PsiOmega}.

Return to \eqref{def:tildeD}, collect the results above to find the second-order term in $\tilde{H}(\epsilon):$
\begin{lemma}\label{LM:tildeD}
  For the terms of order $O(\epsilon^2)$ in $\left( \frac{A+B}{2} \right)^{-\frac{1}{2}}\tilde{H}(\epsilon)\left( \frac{A+B}{2} \right)^{-\frac{1}{2}}:$
  \begin{align}
  \begin{split}
  &\left( \frac{A+B}{2} \right)^{-\frac{1}{2}}\Big[\frac{1}{2}\frac{d^2}{d\epsilon^2}\tilde{H}(\epsilon)\big|_{\epsilon=0}\Big]\left( \frac{A+B}{2} \right)^{-\frac{1}{2}}\\
  =&(1+s)^{-2} \Psi \tilde{\Psi}(s) (E_2 + \Psi E_1) (E_2 + \Psi E_1)^* \tilde{\Psi}(s) \Psi+\tilde{\Psi}(s) \Omega(s) \tilde{\Psi}(s) \Omega^*(s) \tilde{\Psi}(s).
  \end{split}
  \end{align}
  
\end{lemma}

\subsection{Proof of Theorem \ref{THM:secondOrder}}
\begin{proof}
Since $\tilde{H}(\epsilon)$ is the only part containing non-trivial $\epsilon^2$-terms in $H(\epsilon)$ of \eqref{eq:Dvare}, we prove the desired \eqref{eq:Hvare} by the definition of $\tilde{H}(\epsilon)$ and Lemma \ref{LM:tildeD}.

What is left is to prove \eqref{eq:iff}. 

If $E_2+\Psi E_1=0$, then from the definition $\text{Cross}_{A,B}(X,Z)=0.$ 

If $\text{Cross}_{A,B}(X,Z)=0,$ then both integrands in its definition must be zero because they are positive semi-definite for any $s$. For the first one, when $s>0$ is large, it is of the form $s^{-\frac{1}{2}}(1+s)^{-2} (E_2+\Psi E_1) (E_2+\Psi E_1)^{*}+O(s^{-3})$. Hence if $\text{Cross}_{A,B}(X,Z)=0$ then we must have $E_2+\Psi E_1=0.$
\end{proof}


\section{Lieb's concavity theorem}\label{sec:convexity}
We are interested in the following function
\begin{equation}\label{def:ft}
   h(t) := \text{Tr} \big(K^* (tA_1 + (1-t)A_2)^q K (t\overline{B_1} + (1-t)\overline{B_2})^r\big)
\end{equation} where $0 \leq t \leq 1$, $q,r\in \mathbb{R}$ and \begin{equation}\label{eq:qrPair}
q\geq 0;\ \ r\geq 0;\ \text{and}\ p:=q+r\leq 1,
\end{equation}
$A_1$ and $A_2$ are $N\times N$ positive definite matrices, and $B_1$ and $B_2$ are $M\times M$ positive definite matrices, $\overline{B}$ is defined as the complex conjugate of $B,$ $K$ is a $N\times M$ matrix.

In \cite{LIEB1973267}, Lieb proved that $h$ is concave, see also \cite{Epstein1973, carlen2010trace}. 

Here, we are interested in finding an explicit form for $h''$, see Theorem \ref{MainTheorem:concavity} below.

The following well-known result links concavity/convexity of any function $g$ to its second-order derivative $g''$. The identity \eqref{eq:concavity} will be used often.
\begin{lemma}\label{LM:concavity}
Suppose that $g:[0,1]\rightarrow \mathbb{R}$ is smooth function. For $t\in [0,1]$, 
\begin{equation}\label{eq:concavity}
g(t) =(1-t)g(0) + tg(1) -t \int_{0}^{1} \int_{t\lambda}^{\lambda} g''(\sigma) \, d\sigma \, d\lambda.
\end{equation}

It is concave (convex) if and only if $g'' \leq 0$ ($\geq 0$), respectively.
\end{lemma}
\begin{proof}
Suppose that \eqref{eq:concavity} holds. It implies that $g$ is concave (convex) if and only if $g''(t) \leq 0$ ($\geq 0$), respectively.

What is left is to prove \eqref{eq:concavity}. For a fixed $t\in [0,1]$, define a function $\phi:[0,1]\rightarrow \mathbb{R}$ by
\[
\phi(b) := g(tb) - (1-t)g(0) - tg(b).
\]
$\phi$ enjoys the following properties:
\begin{equation}\label{eq:initial}
\phi(0) = \phi'(0)=0.
\end{equation}

Its first order derivative takes the form.
    \[
    \phi'(b) = t[g'(tb) - g'(b)] = -t \int_{tb}^{b} g''(\sigma) \, d\sigma.
    \]
Hence, by \eqref{eq:initial},
\[
\phi(b) = \int_{0}^{b} \phi'(\lambda) \, d\lambda = -t \int_{0}^{b} \int_{t\lambda}^{\lambda} g''(\sigma) \, d\sigma \, d\lambda.
\] 

By setting $b=1$ and using the definition of $\phi$, we obtain the desired \eqref{eq:concavity}.

\end{proof}

We continue to derive an exact expression for $h^{''}$.

We use the ideas of Ando in \cite{ANDO1979203} to rewrite the expression: For any $N\times M$-matrix $K$, there exists a $NM-$dimensional vector $V_{K}$ such that,
\begin{equation}\label{eqn:represen}
    \text{Tr}(K^* A K \overline{B}) = \langle V_K, \big(A \otimes B\big)  V_K \rangle.
\end{equation}
Thus, $h$ in \eqref{def:ft} takes a new form
\begin{equation}\label{eq:htNew}
    h(t) =\langle V_K, H_{q,r}(t) V_K \rangle
\end{equation} 
where $H_{q,r}$ is a matrix-valued function defined as
\begin{equation*}
    H_{q,r}(t) := \big(tA_1 + (1-t)A_2\big)^q \otimes \big(tB_1 + (1-t)B_2\big)^r.
\end{equation*}
Ando proved in \cite{ANDO1979203} that the function $H_{q,r}$ is concave under the condition \eqref{eq:qrPair}.

To simplify notations and to make Theorem \ref{THM:secondOrder} applicable we fix a $t_0\in [0,1]$ and define
\begin{equation}\label{eq:ReforConcave}
G_{q,r}(\epsilon):=  (\Psi + \epsilon V)^q \otimes (\Phi + \epsilon W)^r=H_{q, r}(t_0 + \epsilon)
\end{equation}
where $\epsilon\in \mathbb{R}$ and $|\epsilon| \ll 1$, and $\Psi,$ $\Phi$, $V$ and $W$ are Hermitian matrices defined as
\begin{align*}
    \Psi &:= t_0 A_1 + (1 - t_0) A_2>0, \\
    \Phi &:= t_0 B_1 + (1 - t_0) B_2>0 ,\\
    V &:= A_1 - A_2 ,\\
    W &:= B_1 - B_2.
\end{align*}

Now, we want to find $G_{q, r}''(0)$, which implies the desired $H^{''}_{q,r}(t_0)$ through \eqref{eq:ReforConcave}, and furthermore, it implies the desired $h^{''}$ through \eqref{eq:htNew}.

We begin by examining the cases $(q, r) = (p, 0)$ and $(0, p)$, as the methods for analyzing them differ. Here
\begin{align}
\begin{split}\label{eqn:base}
G_{p, 0} =& (\Psi + \epsilon V)^p \otimes I,\\
G_{0, p} =& I \otimes (\Phi + \epsilon W)^p.
\end{split}
\end{align}

To understand them, we study the general case $(L+\epsilon F)^{q}$, where $q\in [0,1]$, $L$ is positive definite, $F$ is Hermitian.

The result is the following:
\begin{lemma}\label{LM:base}
Suppose that $L$ is positive definite, $F$ is Hermitian, $q\in [0,1]$. 
And suppose that $\epsilon$ is a real scalar satisfying $|\epsilon|\ll 1$.

The following results hold:
  \begin{equation}\label{eqn:Vp1}
    (L + \epsilon F)^q = L^q + \epsilon F_{q,L} -\epsilon^2 K_{q,L,F}+O(\epsilon^3)
\end{equation}
where $K_{q,L,F}$ is a positive semi-definite matrix defined as
\begin{equation*}
    K_{q,L,F} = \frac{\sin(\pi q)}{\pi} \int_{0}^{+\infty} t^{q} (t + L)^{-1} F\ (t + L)^{-1} F \ (t + L)^{-1} \, dt.
\end{equation*}

For $F_{q,L}$, there exists a unique Hermitian matrix $\widetilde{F}_{L}$ such that
\begin{equation}
    \lim_{n \to \infty} 2^n F_{1/2^n,L} = \widetilde{F}_{L}.
\end{equation}
and for any $q\in [0,1]$, $F_{q,L}$ takes the form
\begin{equation}\label{def:F1q}
    F_{q,L} = \int_0^q L^s \widetilde{F}_{L} L^{q-s}\ ds.
\end{equation}

When $q=0$ or 1,
\begin{align}\label{eq:q01}
K_{q,L,F}=  F_{0,L}=0;\ \text{and}\ 
F_{1,L}=F.
\end{align}
\end{lemma}
This lemma will be proved in Section \ref{sec:Base}.

Returning to \eqref{eqn:base}, we apply the results above to find
\begin{align}\label{eq:baseEqui}
\begin{split}
G_{p,0} =& \Psi^p \otimes I + \epsilon V_{p,\Psi} \otimes I - \epsilon^2 \text{Source}_{p,0} + O(\epsilon^3),\\
G_{0,p} =& I \otimes \Phi^p + \epsilon I \otimes W_{p,\Phi} - \epsilon^2 \text{Source}_{0,p} + O(\epsilon^3)
\end{split}  
\end{align}
where  $V_{p,\Psi}$ and $W_{p,\Phi}$ are defined in the sam fashion as \eqref{def:F1q}, and
\begin{align*}
\text{Source}_{p,0} :=& K_{p,\Psi,V} \otimes I,\\
\text{Source}_{0,p} :=&  I \otimes K_{p,\Phi,W}.
\end{align*}

The following result will significantly simplify our consideration for joint convexity in Section \ref{sec:relative}:
by \eqref{eq:q01}
\begin{align}\label{eq:p1}
\begin{split}
\text{Source}_{1,0} =& \text{Source}_{0,1}=0, \\ 
V_{0,\Psi} =& W_{0,\Phi} = 0,\\
V_{1,\Psi} =& V\ \text{and}\ W_{1,\Phi} = W.
\end{split}
\end{align}

We continue to prepare for stating the main results. 

For any pair
\begin{equation}\label{eq:pairQR}
(q, r) = p\ ( \frac{l}{2^{k}}, 1 - \frac{l}{2^{k}} )
\end{equation}
with $1 \le l < 2^{k}$ being an odd integer, and $k\geq 1$ being an integer, we define a positive semi-definite matrix
$\text{Source}_{q,r}$ through $\text{Cross}_{A,B}(X,Z)$ of \eqref{def:Cross}. Here we set
\begin{align}\label{eq:ABXZ}
\begin{split}
X &= V_{q_- ,\Psi} \otimes \Phi^{r_-} + \Psi^{q_-} \otimes W_{r_-, \Phi} ,\\
Z &= V_{q_+, \Psi} \otimes \Phi^{r_+} + \Psi^{q_+} \otimes W_{r_+, \Phi},\\
A &= \Psi^{q_-} \otimes \Phi^{r_-} ,\\
B &= \Psi^{q_+} \otimes \Phi^{r_+},
\end{split}
\end{align}
where, $V_{q_\pm ,\Psi}$ and $W_{r_\pm, \Phi}$ are defined in the same way as $F_{q,L}$ in \eqref{def:F1q}, and the pairs $(q_{\pm},r_{\pm})$ are defined as
\begin{align*}
(q_+, r_+) :=& p \ \big( \frac{l+1}{2^{k}},\  1 - \frac{l+1}{2^{k}}  \big),\\ 
(q_{-}, r_{-}) :=& p\ \big( \frac{l-1}{2^{k}},\   1 - \frac{l-1}{2^{k}} \big).
\end{align*}

We define, for the matrices $A$, $B$, $X$ and $Z$ above,
\begin{align}\label{def:sourceQR}
\text{Source}_{q,r}
    :=\text{Cross}_{A,B}(X,Z).
\end{align}

When $(q,r)=(p,0)$ or $(0,p)$, $\text{Source}_{q,r}$ are defined in \eqref{eq:baseEqui}.

To estimate $\text{Source}_{q,r}$ we need the following result: recall the norm $\|\cdot\|$ defined in \eqref{def:Norm},
\begin{lemma}
 For the pair $(q,r) $ in \eqref{eq:pairQR}, as $k$ increases, $\text{Source}_{q,r}$ vanishes rapidly: For some constant $C=C(\Psi,\Phi,V,W)$,
  \begin{align}\label{eq:estOmegaQR}
    \|\text{Source}_{q,r}\|\leq & C 2^{-2k},
  \end{align}
\end{lemma}
\begin{proof}
Here we need to study the definition of $\text{Cross}_{A,B}(X,Z)$ in \eqref{def:Cross}, with $A$, $B$, $X$ and $Z$ defined in \eqref{eq:ABXZ}. 
$\text{Cross}_{A,B}(X,Z)$, especially the part defined in \eqref{def:PsiE1E2}, depends quadratically on
\[
E_2 + E_1 \Psi = \left( \frac{A+B}{2} \right)^{-\frac{1}{2}} \left[ \frac{X-Z}{2} - \frac{X+Z}{2} \left( \frac{A+B}{2} \right)^{-1} \frac{A-B}{2} \right].
\]

Compute directly to find, for some positive constant $C=C(\Psi,\Phi,V,W)$
\begin{align*}
\| V_{q_+ ,\Psi} - V_{q_- ,\Psi} \| +
\| W_{r_+ ,\Phi} - W_{r_- ,\Phi} \| +
\| \Psi^{q_+} - \Psi^{q_-} \| + 
\| \Phi^{r_+} - \Phi^{r_-} \| \leq C \, 2^{-k}
\end{align*} and hence
\[
\|E_2 + E_1 \Psi\|\leq C \, 2^{-k}.
\]

Since the dependence on $E_2 + E_1 \Psi$ is quadratic, we obtain the desired result \eqref{eq:estOmegaQR} after some straightforward computation.

\end{proof}

We continue to prepare for stating our results.

For any $\delta \in \mathbb{R}$, we define a positive definite matrix 
\begin{equation}\label{def:Hdelta}
  K_{\delta} := \Psi^{\delta} \otimes \Phi^{-\delta}.
\end{equation}
Obviously, as $\delta \to 0$, $K_{\delta}\rightarrow I\otimes I$. To measure the convergence in the norm defined in \eqref{def:Norm}, we have, for some $C=C(\Psi,\Phi)$,
\begin{equation}\label{eq:HDelta}
\| K_{\delta} - I \otimes I \| \le C |\delta|.
\end{equation}

We define a linear operator $D_{\delta}$ for any $\delta\in \mathbb{R}$, such that, for any matrix $F$,
\begin{align}\label{def:Q+-}
  D_{\delta}(F):=\int_{0}^{\infty}e^{-s K_{\delta }}  F e^{-s K_{\delta }} \ ds.
\end{align}
To facilitate later discussions, we need an equivalent form for $D_{\delta}(F)$.
\begin{lemma}\label{LM:DdeltaM}
$D_{\delta}(F)$, defined in \eqref{def:Q+-}, is the unique solution to the following Sylvester equation
\begin{equation}\label{eqn:Q+-}
  K_{\delta}Q+QK_{\delta}=F.
\end{equation}

If $F\geq 0$ ( $>0$), then $D_{\delta}(F)\geq 0$ ($>0$). 

When $|\delta|\ll  1$, $D_{\delta}(F)$ is almost half of $F$: For some constant $C=C(\Psi,\Phi)$, 
\begin{equation}\label{eq:half}
  \|D_{\delta}(F)-\frac{1}{2}F\|\leq C \|F\|  |\delta|.
\end{equation}

\end{lemma}
This lemma will be proved in Sections \ref{sec:LMDdelta} and \ref{sec:half}.

We continue to prepare for stating our main results.

For a fixed pair
\begin{equation}\label{eq:fixPair}
(q, r) =p\ (  \frac{l}{2^k}, 1 - \frac{l}{2^k} ), \ \text{where}\ k,\ l\ \text{are positive integers}; \ l<2^k\ \text{is odd},
\end{equation} we will show that, in \eqref{eq:expreFPrime} below, part of
$G^{''}_{q, r}(0)$ will be generated by $\text{Source}_{q_0, r_0}$ if
\begin{equation}\label{eq:q0r0}
(q_0, r_0) = p\ (  \frac{l_0}{2^{k_0}}, 1 - \frac{l_0}{2^{k_0}} )
\end{equation}
satisfies the conditions:
\begin{align}\label{eq:closeness}
\begin{split}
\left| \frac{l_0}{2^{k_0}} - \frac{l}{2^k} \right| & < 2^{-k_0};\ k_0<k\  \text{and}\ l_0\ \text{are nonnegative integers },\\
\frac{l_0}{2^{k_0}}\ \text{is of simplified form}:\ & k_0\geq 0;\ l_0 \leq 2^{k_0} \ \text{is odd if} \ k_0\geq 1; \ 0=\frac{0}{2^0};\ 1=\frac{1}{2^0}.
\end{split}
\end{align}

Under the condition \eqref{eq:closeness}, we observe that there exist integers
\begin{equation}\label{eq:increasing}
k_0 < m_1 < m_2 < \dots < m_v < k 
\end{equation}
and $c_j =  1$ or $-1$, $j=1,2,\cdots,v,$ such that
\begin{equation}\label{eq:series}
\frac{l_0}{2^{k_0}} - \frac{l}{2^k} = \sum_{j=1}^{v} c_j 2^{-m_j}. 
\end{equation}
For example, when $(\frac{l}{2^{k}},\ 1-\frac{l}{2^{k}})=(\frac{5}{8},\ \frac{3}{8})$ and $(\frac{l_0}{2^{k_0}},\ 1-\frac{l_0}{2^{k_0}})=(\frac{1}{2},\frac{1}{2})$,
\begin{align}\label{eq:example}
\frac{l_0}{2^{k_0}}-\frac{l}{2^{k}}=-\frac{1}{8}=-\frac{1}{4}+\frac{1}{8}.
\end{align}

Before showing the dependence of $G^{''}_{q,r}(0)$ on $\text{Source}_{q_0,r_0}$ we define a set of a finite sequence of pairs
$$\Omega_{q,q_0} := \Big\{ \Big(\ (m_1,c_1),\ (m_2,c_2), \dots, (m_v,c_v)\Big)\ \Big|\ \eqref{eq:closeness},\ \eqref{eq:increasing} \text{ and } \eqref{eq:series} \text{ are satisfied} .\Big\}$$
Then the following term is useful for constructing $G_{q, r}^{''}(0)$: recall the operator $D_{\delta}$ from \eqref{def:Q+-},
$$Y_{  \frac{l}{2^k}, \frac{l_0}{2^{k_0}}, p}:=\sum_{\Omega_{q,  q_0 }}  D_{c_v p 2^{-m_v}}\Big( \dots D_{c_1 p 2^{-m_1}} \big( \text{Source}_{ q_0,r_0 }  \big)\Big)\geq 0.$$
To help understand this, we use the example \eqref{eq:example},
\begin{equation}\label{eq:exampleIter}
Y_{\frac{5}{8},\frac{1}{2},p}=D_{-\frac{p}{8}} (\text{Source}_{\frac{p}{2},\frac{p}{2}})+D_{\frac{p}{8}}D_{-\frac{p}{4}}(\text{Source}_{\frac{p}{2},\frac{p}{2}})\geq 0.
\end{equation}

We are ready to state the main result in this section, whose \eqref{eq:Lieb} implies Lieb's concavity theorem: recall the norm $\|\cdot \|$ in \eqref{def:Norm}, the fixed $t_0$ in \eqref{eq:ReforConcave}, and that $\text{Source}_{q,r}$ satisfies \eqref{eq:estOmegaQR},
\begin{theorem}\label{MainTheorem:concavity}
 For the fixed pairs $(q,r)$ described in \eqref{eq:fixPair}, $G^{''}_{q,r}(0)$ is generated by $\text{Source}_{q_0,r_0}$ with $(q_0,r_0)$ described in \eqref{eq:q0r0} satisfying \eqref{eq:closeness}, specifically,   
\begin{equation}\label{eq:expreFPrime}
0\geq\frac{1}{2}H^{''}_{q,r}(t_0)= \frac{1}{2}G^{''}_{q,r} (0)= -\sum_{\frac{l_0}{2^{k_0}}\ \text{satisfying} \ \eqref{eq:closeness}
} Y_{  \frac{l}{2^k},  \frac{l_0}{2^{k_0}}, p }.
\end{equation}

There exists some constant $C=C(\Psi,\Phi, V, W)>0$, such that
\begin{equation}\label{eq:estiX}
\Big\|  Y_{\frac{l}{2^k}, \frac{l_0}{2^{k_0}}, p }-\big(1-|l_0-\frac{l}{2^{k-k_0}}|\big) \text{Source}_{q_0,r_0}\Big\|\leq C 2^{-k_0} \Big\|\text{Source}_{q_0,r_0}\Big\|.
\end{equation} 

For any nonnegative $q$ and $r$ satisfying $$q+r=p<1$$ and $q$ is not of the form $p\frac{l}{2^{k}}$, there exists a unique sequence of increasing positive integers $l_1<l_2<\cdots,$ such that
$$q=p\sum_{k=1}^{\infty}2^{-l_k},$$
the following equality and limit hold: by \eqref{eq:estiX} and the estimate for $\text{Source}_{q,r}$ in \eqref{eq:estOmegaQR}
\begin{equation}\label{eq:limitDeri}
  G^{''}_{q,r}(0)=H^{''}_{q,r}(t_0)=\lim_{n\rightarrow \infty}H^{''}_{p\sum_{k=1}^{n} 2^{-l_k}, \ p\big(1-\sum_{k=1}^{n} 2^{-l_k}\big)}(t_0)\leq 0.
\end{equation} 

The function $h$ defined in \eqref{def:ft} is concave: by \eqref{eq:htNew} and the results above
\begin{align}
  h''(t)=\big\langle V_K,\ H_{q,r}^{''}(t) V_{K}\big\rangle\leq 0,
\end{align} and by \eqref{eq:concavity}, when $t\in (0,1),$
\begin{equation}\label{eq:Lieb}
  h(t)-(1-t)h(0)-t h(1)=-t\int_{0}^{1}\int_{t\lambda}^{\lambda} \big\langle V_K,\ H_{q,r}^{''}(\sigma) V_{K}\big\rangle\ d\sigma d\lambda\geq 0.
\end{equation}
\end{theorem}
The theorem will be reformulated in subsection \ref{subsec:concave} and proved in subsection \ref{sub:concavity}.

\subsection{Reformulation of Theorem \ref{MainTheorem:concavity}}\label{subsec:concave}
For fixed nonnegative constants $p,q,r$ such that $p \in [0, 1]$ and  $$q + r = p,$$ by standard perturbation theory, there exists $V_{q,\Psi}$, $W_{r,\Phi}$ and $ Q_{q,r}$ such that
\begin{equation}\label{eq:detailEpsilon}
  G_{q, r}(\epsilon) = \Psi^q \otimes \Phi^r+\epsilon V_{q,\Psi}\otimes \Phi^r+\epsilon \Psi^q \otimes W_{r,\Phi}-\epsilon^2 Q_{q,r}+O(|\epsilon|^3). 
\end{equation} We already found the desired $V_{q,\Psi}$ and $W_{r,\Phi}$ in \eqref{eqn:Vp1}. 

The focus is to find an explicit expression for $$Q_{q,r}=-\frac{1}{2}G_{q, r}''(0)\geq 0.$$ The following identity, implied by \eqref{eq:detailEpsilon}, will be used often:
\begin{align}\label{eq:pmEpsilon}
    \frac{1}{2}\sum_{\pm}G_{q, r}(\pm\epsilon)=\Psi^q \otimes \Phi^r-\epsilon^2 Q_{q,r}+O(|\epsilon|^3).
\end{align}

We will use the ideas in \cite{ANDO1979203}, see also \cite{carlen2010trace}, and find the desired $G_{q, r}''(0)$ through iteration and approximation. Specifically, when $(q,r)=p(\frac{l}{2^n},1-\frac{l}{2^n})$, where $n,l$ are nonnegative integers and $l\leq 2^n$, we will find the desired $Q_{q,r}$ by iterating $n$. Otherwise, we will approximate it using the known cases.

To initiate the iteration, we choose the base cases to be $(q, r) = (p, 0)$ and $(0, p)$, where
\begin{align*}
G_{p, 0} =& (\Psi + \epsilon V)^p \otimes I,\\
G_{0, p} =& I \otimes (\Phi + \epsilon W)^p.
\end{align*}
For these cases, we take results from \eqref{eq:baseEqui}:
\begin{align}
\begin{split}
  Q_{p,0}=&\text{Source}_{p,0},\\
  Q_{0,p}=&\text{Source}_{0,p}.
\end{split}
\end{align}

For the next step of iteration, suppose that we find the desired $Q_{q,r}$ when $(q,r)\in K_{I}$, which is a set of pairs defined as, for a fixed integer $I \geq 0$,
\begin{equation}\label{eq:cases}
K_{I}:=\Big\{  (q, r) = p\ \big(  \frac{l}{2^k},\ 1 - \frac{l}{2^k}  \big)\ \Big| \ 0 \leq k \leq I;\ 0 \leq l \leq 2^k; \ k,\ l \in \mathbb{Z} \Big\}.
\end{equation}

Based on this, we consider cases not included in $K_{I}$.
\begin{equation}\label{eq:oddL}
(q, r) = p\ (  \frac{l}{2^{I+1}},   1 - \frac{l}{2^{I+1}}  ),\ \text{and}\ l\in (0, 2^{I+1})\ \text{is an odd integer}.
\end{equation}
We will show that part of the desired $Q_{q,r}$ is generated by $Q_{q_{\pm},r_{\pm}}$ with $(q_{\pm},r_{\pm})\in K_{I}$ defined as
\begin{align*}
(q_+, r_+) :=& p\ (  \frac{l+1}{2^{I+1}},\ 1 - \frac{l+1}{2^{I+1}}  ),\\ 
(q_{-}, r_{-}) :=&p\ (  \frac{l-1}{2^{I+1}},\   1 - \frac{l-1}{2^{I+1}}  ).
\end{align*} Because $l - 1$ and $l + 1$ are even integers, these two cases
are indeed included in the set $K_I$.

Our idea is stimulated by \cite{ANDO1979203}: the definition of geometric mean $M_0$ in \eqref{def:GeoMean} implies
\begin{equation*}
  G_{q, r}(\epsilon)=M_0\big(G_{q_{+}, r_{+}}(\epsilon),\ G_{q_{-}, r_{-}}(\epsilon)\big).
\end{equation*}
Then, we use the identity
\begin{align}\label{eq:difference1}
\begin{split}
   \frac{1}{2} \epsilon^2 G^{''}_{q,r}(0)+O(\epsilon^3) =&\frac{1}{2}\sum_{\pm}G_{q, r}(\pm \epsilon)-G_{q, r}(0)\\
     =&\frac{1}{2}\sum_{\pm}M_0\big(G_{q_{+}, r_{+}}(\pm\epsilon),\ G_{q_{-}, r_{-}}(\pm\epsilon)\big)-M_0\big(G_{q_{+}, r_{+}}(0),\ G_{q_{-}, r_{-}}(0)\big)\\
     =&\Xi_1+\Xi2
\end{split}
\end{align} to link the desired $G^{''}_{q,r}(0)$ to $G^{''}_{q_{\pm},r_{\pm}}(0)=-2Q_{q_{\pm},r_{\pm}}$. Here $\Xi_1$ and $\Xi_2$ are Hermitian matrices defined as
\begin{align*}
\begin{split}
  \Xi_1:=&M_0 \Big( \frac{1}{2} \sum_{\pm}G_{q_+, r_+}(\pm\epsilon)  , \ \frac{1}{2}  \sum_{\pm}G_{q_-, r_-}(\pm\epsilon)   \Big)-M_0\big(G_{q_{+}, r_{+}}(0),\ G_{q_{-}, r_{-}}(0)\big),\\
  \Xi_2:=&\frac{1}{2}\sum_{\pm}M_0\big(G_{q_{+}, r_{+}}(\pm\epsilon),\ G_{q_{-}, r_{-}}(\pm\epsilon)\big)-M_0 \Big( \frac{1}{2} \sum_{\pm}G_{q_+, r_+}(\pm\epsilon)  , \ \frac{1}{2}  \sum_{\pm}G_{q_-, r_-}(\pm\epsilon)   \Big).
\end{split}
\end{align*}

Before studying $\Xi_1$ and $\Xi_2$ we recall the definition of the operator $D_{\delta}$ from \eqref{def:Q+-}, recall the definition of $\text{Source}_{q,r}$ from \eqref{eq:ABXZ} and \eqref{def:sourceQR}, and define a positive scalar $$\delta_I:=p 2^{-I}.$$

$\Xi_1$ and $\Xi_2$ satisfy the following equations:
\begin{lemma}
For $\Xi_1,$
\begin{align}\label{eq:step1}
\begin{split}
\Xi_1=
&-\epsilon^2 D_{\delta_{I} }(Q_{q_+, \, r_+})-\epsilon^2 D_{-\delta_{I} }(Q_{q_-, \, r_-})+ O(\epsilon^3).
\end{split}
\end{align}

For $\Xi_2,$
\begin{align}\label{eq:step2}
\begin{split}
\Xi_2
=-\epsilon^2 \text{Source}_{q,r}+O(\epsilon^3).
\end{split}
\end{align}
\end{lemma}
These two equations will be proved in subsubsections \ref{subsub:step1} and \ref{subsub:step2}, respectively.

Here we present the ideas.

To control $\Xi_1$ we observe that, by \eqref{eq:pmEpsilon},
\begin{align}
& M_0 \left( G_{q_+, r_+}(0), \ G_{q_-, r_-}(0) \right) \nonumber\\
=& M_0 \Big( \frac{1}{2} \sum_{\pm}G_{q_+, r_+}(\pm\epsilon)  + \epsilon^2 Q_{q_+, r_+}, \ \frac{1}{2} \sum_{\pm}G_{q_-, r_-}(\pm\epsilon)  + \epsilon^2 Q_{q_-, r_-} \Big)+ O(\epsilon^3).\label{eq:TwoSecondOrder}
\end{align}
We aim to move the terms $\epsilon^2 Q_{q_+, r_+}$ and $\epsilon^2 Q_{q_-, r_-}$ outside of $M_0(\cdot,\cdot)$. Even though there are some other terms of order $O(\epsilon^2)$ within $M_0(\cdot,\cdot)$, their contribution is negligible, being of order $O(\epsilon^3)$. This, combined with straightforward calculations, leads to the desired outcome, as detailed in Proposition \ref{prop:step1} within subsubsection \ref{subsub:step1}.

Controlling $\Xi_2$ is straightforward. Based on \eqref{eq:detailEpsilon}, we have
\[
G_{q_{\pm}, r_{\pm}}(\epsilon) = G_{q_{\pm}, r_{\pm}}(0) +  \epsilon V_{q_{\pm},\Psi}\otimes \Phi^r+\epsilon \Psi^q\otimes W_{r_{\pm},\Phi} - \epsilon^2 Q_{q_{\pm}, r_{\pm}} + O(\epsilon^3).
\]
Due to cancellations, the contributions from the terms $\epsilon^2 Q_{q_{\pm}, r_{\pm}}$ are negligible, of order $O(\epsilon^3)$. This makes the problem similar to that of geometric means, as examined in Theorem \ref{THM:secondOrder}. The detailed analysis is presented in Proposition \ref{prop:p1p2} of subsubsection \ref{subsub:step2}.

We are ready to state the main result for the second step of iteration:

\begin{proposition}\label{prop:Iteration} 
For the $(q,r)$ in \eqref{eq:oddL}, the term $Q_{q,r}=-\frac{1}{2}G^{''}_{q,r}(0)$ has three parts:  
  \begin{equation}\label{eq:OmegaQR}
  Q_{q,r}=D_{\delta_{I} }(Q_{q_+, \, r_+})+D_{-\delta_{I} }(Q_{q_-, \, r_-})+\text{Source}_{q,r}
  \end{equation} and 
  $D_{\delta_{I} }(Q_{q_+, \, r_+})\geq 0$ and $D_{-\delta_{I}}(Q_{q_-, \, r_-})\geq 0$ are generated by $Q_{q_{+},r_{+}}$ and $Q_{q_{-},r_{-}}$, $\text{Source}_{q,r}$ is defined in \eqref{eq:ABXZ} and \eqref{def:sourceQR}.

  $D_{\delta_{I} }(Q_{q_+, \, r_+})$ and $D_{-\delta_{I}}(Q_{q_-, \, r_-})$ are about half of $Q_{q_{+},r_{+}}$ and $Q_{q_{-},r_{-}}$ respectively: for some $C=C(\Psi,\Phi, V, W)$,
  \begin{align}\label{eq:qqrhalf}
    \big\|D_{\delta_{I} }(Q_{q_+, \, r_+})-\frac{1}{2}Q_{q_{+},r_{+}}\big\|+\big\|D_{-\delta_{I}}(Q_{q_-, \, r_-})-\frac{1}{2}Q_{q_{-},r_{-}}\big\|\leq & C 2^{-I}.
  \end{align}
\end{proposition}
\begin{proof}

\eqref{eq:difference1}, \eqref{eq:step1} and \eqref{eq:step2} imply the desired \eqref{eq:OmegaQR}.

\eqref{eq:qqrhalf} is implied by \eqref{eq:half}.

The proof of Proposition \ref{prop:Iteration} is complete, provided that we prove \eqref{eq:step1} and \eqref{eq:step2}.
\end{proof}

\subsubsection{Proof of \eqref{eq:step1}}\label{subsub:step1}
To use the same notations as Theorem \ref{THM:secondOrder}, we set
\begin{align}\label{def:ABQ}
\begin{split}
    A&:=\Psi^{q_{+}}\otimes \Phi^{r_{+}},\\
    B&:=\Psi^{q_{-}}\otimes \Phi^{r_{-}},
\end{split}
\end{align} Obviously $A$ and $B$ commute: $AB=BA.$
And define
\begin{align*}
    Q_1&:=Q_{q_+,r_+},\\
    Q_2&:=Q_{q_-,r_-}.
\end{align*}
And let $P_1$ and $P_2$ be the unique matrices to make
\begin{align}\label{def:P1P2}
\begin{split}
    \frac{1}{2} \sum_{\pm}G_{q_+, r_+}(\pm\epsilon) &=\Psi^{q_{+}}\otimes \Phi^{r_{+}}+\epsilon^2 P_1+O(\epsilon^3),\\
    \frac{1}{2} \sum_{\pm}G_{q_-, r_-}(\pm\epsilon)&=\Psi^{q_{-}}\otimes \Phi^{r_{-}}+\epsilon^2 P_2+O(\epsilon^3).
\end{split}
\end{align}

The next result is slightly more general. \eqref{eq:Q1Q2} implies the desired \eqref{eq:step1} because $D_{\delta}(F)$ admits two equivalent formulations \eqref{def:Q+-} and \eqref{eqn:Q+-}.
\begin{proposition}\label{prop:step1}

Suppose that $A, B$ are positive definite, and they commute $$AB = BA.$$ Suppose that $P_1, P_2, Q_1, Q_2$ are Hermitian matrices, and $\epsilon \in \mathbb{R}$ satisfies $|\epsilon| \ll 1$.

Then,
\begin{align}\label{eq:P1P2Q1Q2}
\begin{split}
    &M_0(A + \epsilon^2 P_1 + \epsilon^2 Q_1, \ B + \epsilon^2 P_2 + \epsilon^2 Q_2) - M_0(A + \epsilon^2 P_1, \ B + \epsilon^2 P_2) \\
    =& M_0(A + \epsilon^2 Q_1, \ B) + M_0(A, \ B + \epsilon^2 Q_2) - 2M_0(A, B) + O(\epsilon^3).
\end{split}
\end{align}
For the terms on the right-hand side,
\begin{align}\label{eq:Q1Q2}
\begin{split}
M_0(A + \epsilon^2 Q_1, \ B) - M_0(A, B) =& \epsilon^2 \tilde{Q}_1 + O(\epsilon^4),\\
M_0(A, \ B + \epsilon^2 Q_2) - M_0(A, B) =&\epsilon^2 \tilde{Q}_2 + O(\epsilon^4),
\end{split}
\end{align}
where $\tilde{Q}_1$ and $\tilde{Q}_2$ are the unique solutions to the equations:
\[
(B^{-\frac{1}{2}} A B^{-\frac{1}{2}})^{\frac{1}{2}} \tilde{Q}_1 + \tilde{Q}_1 (B^{-\frac{1}{2}} A B^{-\frac{1}{2}})^{\frac{1}{2}} = Q_1,
\]
\[
(A^{-\frac{1}{2}} B A^{-\frac{1}{2}})^{\frac{1}{2}} \tilde{Q}_2 + \tilde{Q}_2 (A^{-\frac{1}{2}} B A^{-\frac{1}{2}})^{\frac{1}{2}} = Q_2.
\]
\end{proposition}
\begin{proof}
We start with proving \eqref{eq:P1P2Q1Q2}.

Since \eqref{def:formuMeans} implies that the left-hand side of \eqref{eq:P1P2Q1Q2} is a smooth function of $\epsilon^2$, and because we are focusing on the $\epsilon^2$ terms, we only need to consider the linear contributions of $\epsilon^2 P_j$ and $\epsilon^2 Q_j$, for $j = 1, 2$. Consequently, the linear contributions from $\epsilon^2 P_1$ and $\epsilon^2 P_2$ cancel out, and the interaction between $\epsilon^2 Q_1$ and $\epsilon^2 Q_2$ is of order $O(\epsilon^4)$. By these we prove the desired result. 

The detailed proof, though straightforward, is somewhat lengthy, so we omit it.
    
Next, we prove the first equation in \eqref{eq:Q1Q2}; the proof of the second equation is almost identical, so we omit it.

Let $$H := M_0(A + \epsilon^2 Q_1, B),$$ i.e., $H$ is the maximal Hermitian matrix to make:
\[
\Pi := \begin{bmatrix} A + \epsilon^2 Q_1 & H \\ H & B \end{bmatrix} \geq 0
\]
Multiply both sides by $\begin{bmatrix} B^{-\frac{1}{2}} & 0 \\ 0 & B^{-\frac{1}{2}} \end{bmatrix}$ to find:
\begin{align*}
    0 &\leq \begin{bmatrix} B^{-\frac{1}{2}} & 0 \\ 0 & B^{-\frac{1}{2}} \end{bmatrix} \Pi \begin{bmatrix} B^{-\frac{1}{2}} & 0 \\ 0 & B^{-\frac{1}{2}} \end{bmatrix} = \begin{bmatrix} B^{-\frac{1}{2}} A B^{-\frac{1}{2}} + \epsilon^2 B^{-\frac{1}{2}} Q_1 B^{-\frac{1}{2}} & B^{-\frac{1}{2}} H B^{-\frac{1}{2}} \\ B^{-\frac{1}{2}} H B^{-\frac{1}{2}} & I \end{bmatrix}
\end{align*}
Since the matrix on the right hand side is positive semi-definite, we must have:
\begin{equation}\label{def:K}
B^{-\frac{1}{2}} H B^{-\frac{1}{2}} \leq (B^{-\frac{1}{2}} A B^{-\frac{1}{2}} + \epsilon^2 B^{-\frac{1}{2}} Q_1 B^{-\frac{1}{2}})^{\frac{1}{2}}=:K.
\end{equation}
And since $B$ is positive definite and $H$ is the maximal one,
\begin{equation}\label{eq:expandH}
H = B^{\frac{1}{2}} K B^{\frac{1}{2}} .
\end{equation}

To understand $H$, we need to study $K$ defined in \eqref{def:K}. By standard perturbation expansion, there exists a Hermitian matrix $Y_1$ such that
\begin{align}\label{eq:expand2}
    K&=  (B^{-\frac{1}{2}} A B^{-\frac{1}{2}})^{\frac{1}{2}} + \epsilon^2 Y_1 + O(\epsilon^4).
\end{align}
Take a square on both sides and use \eqref{def:K} to find
\begin{equation*}
  \left[ (B^{-\frac{1}{2}} A B^{-\frac{1}{2}})^{\frac{1}{2}} + \epsilon^2 Y_1 + O(\epsilon^4) \right]^2=K^2=B^{-\frac{1}{2}} A B^{-\frac{1}{2}} + \epsilon^2 B^{-\frac{1}{2}} Q_1 B^{-\frac{1}{2}}.
\end{equation*}
This yields an equation for $Y_1$:
\[
(B^{-\frac{1}{2}} A B^{-\frac{1}{2}})^{\frac{1}{2}} Y_1 + Y_1 (B^{-\frac{1}{2}} A B^{-\frac{1}{2}})^{\frac{1}{2}} = B^{-\frac{1}{2}} Q_1 B^{-\frac{1}{2}}.
\]

Return to \eqref{eq:expandH} and \eqref{eq:expand2}. Since $A$ and $B$ commute, $\tilde{Q}_1 := B^{\frac{1}{2}} Y_1 B^{\frac{1}{2}}$ satisfies the desired equation \eqref{eq:Q1Q2}:
\[
(B^{-\frac{1}{2}} A B^{-\frac{1}{2}})^{\frac{1}{2}} \tilde{Q}_1 + \tilde{Q}_1 (B^{-\frac{1}{2}} A B^{-\frac{1}{2}})^{\frac{1}{2}} = Q_1.
\]

\end{proof}

\subsubsection{Proof of \eqref{eq:step2}}\label{subsub:step2}
To simplify the notations, we let $A$, $B$, $P_1$, and $P_2$ be the same to those defined in \eqref{def:ABQ} and \eqref{def:P1P2}. We define $X$ and $Z$ as Hermitian matrices:
\begin{align*}
    X:=&V_{q_{+},\Psi}\otimes \Phi^{r_{+}}+ \Psi^{q_{+}}\otimes W_{r_{+},\Phi},\\
    Z:=&V_{q_{-},\Psi}\otimes \Phi^{r_{-}}+ \Psi^{q_{-}}\otimes W_{r_{-},\Phi}.
\end{align*}
By \eqref{eq:detailEpsilon},
\begin{align*}
G_{q_+, r_+}(\epsilon)-\frac{1}{2}  \sum_{\pm}G_{q_+, r_+}(\pm\epsilon)    =&  \epsilon X  + O(\epsilon^3), \\
G_{q_-, r_-}(\epsilon)-\frac{1}{2}  \sum_{\pm}G_{q_-, r_-}(\pm\epsilon)    =&  \epsilon  Z + O(\epsilon^3).
\end{align*}

The next proposition implies the desired \eqref{eq:step2}.
\begin{proposition}\label{prop:p1p2}
Suppose that $A, B$ are positive definite. $P_1$, $P_2$, $X$ and $Z$ are Hermitian. $\epsilon \in \mathbb{R}$ satisfies $|\epsilon| \ll 1$.
Then,
\begin{align}\label{eq:CrossExtra}
\begin{split}
    &\sum_{\pm}M_0(A + \epsilon^2 P_1  \pm \epsilon X, \ B + \epsilon^2 P_2 \pm \epsilon Z)  
    - 2 M_0(A + \epsilon^2 P_1, \ B + \epsilon^2 P_2) \\
    =& -2 \epsilon^2 \text{Cross}_{A,B}(X,Z) + O(\epsilon^3).
\end{split}
\end{align}
\end{proposition}

\begin{proof}
Theorem \ref{THM:secondOrder} implies that
\begin{align}
\begin{split}
&\sum_{\pm}M_0(A \pm \epsilon X, \ B \pm \epsilon Z)  - 2 M_0(A, B)
= -2 \epsilon^2 \text{Cross}_{A,B}(X,Z) + O(\epsilon^3).
\end{split}
\end{align}

The remaining step is to note that, as \eqref{def:formuMeans} implies the left-hand side of \eqref{eq:CrossExtra} is a smooth function of $\epsilon$, the contributions from the terms $\epsilon^2 P_1$ and $\epsilon^2 P_2$ are of order $O(\epsilon^3)$ due to cancellations. The detailed proof is straightforward but tedious, and is therefore omitted.

\end{proof}


\subsection{Proof of Theorem \ref{MainTheorem:concavity}}\label{sub:concavity}
\eqref{eq:expreFPrime} and \eqref{eq:estiX} can be obtained by a standard induction process using the iteration relation \eqref{eq:OmegaQR}, by which all the terms in $Q_{q,r}$ can be traced back to some $\text{Source}_{q_0,r_0}$.

For the induction, we provided example in \eqref{eq:exampleIter} to illustrate the idea of proving \eqref{eq:expreFPrime}. For \eqref{eq:estiX}, we use the same example and apply \eqref{eq:half} to find, for some $C=C(\Psi,\Phi)>0$,
\begin{align}
\|Y_{\frac{5}{8},\frac{1}{2},p}-\frac{3}{4}\text{Source}_{\frac{p}{2},\frac{p}{2}}\|\leq C \|K_{\frac{p}{4}}-I\|.
\end{align}

Based on the example, proving \eqref{eq:expreFPrime} and \eqref{eq:estiX} is straightforward. We choose to skip it.

\eqref{eq:limitDeri}-\eqref{eq:Lieb} were proved in the context.

The proof is complete.

\section{Joint convexity of the relative entropy}\label{sec:relative}
When $A$ and $B$ are positive definite $N\times N$ density matrices, we are interested the joint convexity in the following mapping:
\begin{align}\label{def:relative}
    (A,B)\rightarrow S(A|B) &:= \text{Tr}[A \log_2 A] - \text{Tr}[A \log_2 B].
\end{align} 

We begin with making Theorem \ref{MainTheorem:concavity} applicable.

As in the previous section, we follow Ando's method in \cite{ANDO1979203}, see also \cite{carlen2010trace}: for some $ N^2$-dimensional vector $V_{I}$,
\begin{equation}\label{def:conditional}
S(A|B) = \left\langle V_I, F(A,\overline{B}) V_I \right\rangle
\end{equation}
with $F(A,\overline{B})$ defined as $$F(A, \overline{B}) := A \log_2 A \otimes I - A \otimes \log_2 \overline{B}.$$

To make Theorem \ref{MainTheorem:concavity} applicable, we find 
\begin{equation}\label{limit:FnAB}
    F(A,\overline{B})=-\log_2 e\ \lim_{k\rightarrow +\infty} \frac{F_{k}(A,\overline{B})-F_{0}(A,\overline{B})}{\delta_{k}}
\end{equation}
where $\delta_{k}>0$ is a scalar defined as $$\delta_{k}:=2^{-k},$$ and $F_k(A,\overline{B})$ is a positive definite matrix defined as
$$F_{k}(A,\overline{B}):=A^{1-\delta_k} \otimes \overline{B}^{\delta_k}.$$

We will apply Theorem \ref{MainTheorem:concavity} to derive an expression for $\frac{d^2}{dt^2}F_k\big(tA_1+(1-t)A_2,t\overline{B_1}+(1-t)\overline{B_2}\big)$. This implies the desired $\frac{d^2}{dt^2}F\big(tA_1+(1-t)A_2,t\overline{B_1}+(1-t)\overline{B_2}\big)$ through the identity
\begin{align}\label{eq:FtA}
\begin{split}
    &\frac{d^2}{dt^2} F\big(tA_1+(1-t)A_2,t\overline{B_1}+(1-t)\overline{B_2}\big)\\
    =&-(\log_2 e)\lim_{k\rightarrow +\infty} \frac{1}{\delta_k}\frac{d^2}{dt^2} F_{k}\big(tA_1+(1-t)A_2,t\overline{B_1}+(1-t)\overline{B_2}\big), 
\end{split}
\end{align} where we used that $\frac{d^2}{dt^2}F_{0}\big(tA_1+(1-t)A_2,t\overline{B_1}+(1-t)\overline{B_2}\big)=0.$

To make the application of Theorem \ref{MainTheorem:concavity} transparent, we adopt the same notations by fixing a $t_0 \in [0, 1]$ and denoting
\begin{align}\label{def:PsiPhiVW}
\begin{split}
  \Psi &= t_0 A_1 + (1 - t_0) A_2, \\
    \Phi &= t_0 \overline{B_1} + (1 - t_0) \overline{B_2}, \\
    V &= A_1 - A_2, \\
    W &= \overline{B_1} - \overline{B_2}.
\end{split}
\end{align}
Corresponding to the definition of $G_{q, r}$ in \eqref{eq:ReforConcave}, we define
\begin{align}\label{eq:defGk}
\begin{split}
G_k(\epsilon) :=& (\Psi + \epsilon V)^{1-\delta_k} \otimes (\Phi + \epsilon W)^{\delta_k}\\
=&F_{k}\big((t_0+\epsilon)A_1+(1-t_0-\epsilon)A_2,\ (t_0+\epsilon)\overline{B_1}+(1-t_0-\epsilon)\overline{B_2}\big).
\end{split}
\end{align} 

By the condition \eqref{eq:closeness} and the result \eqref{eq:expreFPrime},  $G^{''}_k(0)$ is generated by the following terms: $\text{Source}_{0,1}$, $\text{Source}_{1,0}$ and 
\begin{equation}
\text{Source}_l:=\text{Source}_{1-\delta_l,\delta_l},\ \text{with}\ l=1,2,\cdots,k.
\end{equation}
To make our problem significantly easier we observe that, by \eqref{eq:p1},
\begin{equation}\label{eq:0source} 
\text{Source}_{0,1} = \text{Source}_{1, 0} = 0 .
\end{equation}

By \eqref{def:sourceQR}, when $l\geq 1$, $\text{Source}_{l }$ is positive semi-definite and defined as, 
\begin{align}
\begin{split}
&\text{Source}_{l }
:= \text{Cross}_{\Psi^{1-\delta_{l-1} } \otimes \ \Phi^{\delta_{l-1} }, \ \Psi \otimes I} 
\left( V_{1-\delta_{l-1},\Psi } \otimes \Phi^{\delta_{l-1} } + \Psi^{1-\delta_{l-1}} \otimes W_{\delta_{l-1},\Phi }, \ V \otimes I  \right),
\end{split}
\end{align} where, recall that $\delta_{l}:=2^{-l},$ and $V_{p,\Psi}$ and $W_{q,\Phi}$ are the coefficients of the $\epsilon-$term in the expansion of $(\Psi+\epsilon V)^p$ and $(\Phi+\epsilon W)^q$ and they are defined in \eqref{eq:baseEqui}.
\eqref{eq:estOmegaQR} implies that, in the norm defined in \eqref{def:Norm}, for some $C=C(\Psi,\Phi, V,W)$,
\begin{equation}\label{eq:sourcel}
\| \text{Source}_{l} \| \leq C 2^{-2l}.
\end{equation}

Before applying \eqref{eq:expreFPrime} in Theorem \ref{MainTheorem:concavity}, we define 
\begin{align}\label{def:Gammak}
\begin{split}
\Gamma_k(t_0) :=& \text{Source}_k + D_{-2^{-k}}(\text{Source}_{k-1})\\
& + D_{-2^{-k}}(D_{-2^{-k+1}}(\text{Source}_{k-2})) + \dots + D_{-2^{-k}}(D_{-2^{-k+1}} \dots (D_{-\frac{1}{4}} \text{Source}_1)).
\end{split}
\end{align} Because of \eqref{eq:0source} we can ignore $\text{Source}_{1,0}$ and $\text{Source}_{0,1}$.
Recall that $D_{\delta}$ is an operator defined in \eqref{def:Q+-}, and $D_{\delta}(T)\geq 0$ if $T\geq 0;$ and we displayed the dependence of $\Gamma_k$ on $t_0$, see \eqref{def:PsiPhiVW}.

The main result for this section is the following: 
\begin{theorem}\label{MainTHM:convex}
By \eqref{eq:expreFPrime} in Theorem \ref{MainTheorem:concavity} and \eqref{eq:0source},
\begin{equation}
  G_{k}^{''}(0)=-2\Gamma_{k}(t_0)\leq 0.
\end{equation}  

There exists a positive semi-definite matrix $\Gamma(t_0)$ such that $2^{k} \Gamma_k(t_0)$ converges rapidly to $\Gamma(t_0)$: for some constant $C=C(\Psi,\Phi, V, W)$,
\begin{equation}\label{eq:limitGamma}
\|2^k \Gamma_k(t_0) - \Gamma(t_0)\|\leq C 2^{-k}.
\end{equation}
By \eqref{eq:FtA} and the second identity in \eqref{eq:defGk},
\begin{equation}\label{eq:GammaT}
2 (\log_2 e)\Gamma(t)  = \frac{d^2}{dt^2} F\big(tA_1+(1-t)A_2, tB_1+(1-t)B_2\big) 
\end{equation} and thus, by \eqref{def:conditional} and \eqref{eq:concavity},
\begin{align}
\begin{split}
  S\big(tA_1+(1-t)A_2\ \big|\ tB_1+(1-t)B_2\big)=&tS\big(A_1\ \big| B_1\big)+(1-t)S\big(A_2\ \big| B_2\big)\\
  &-2(\log_2 e) t\int_{0}^{1} \int^{\lambda}_{t\lambda} \langle V_{I}, \Gamma(\sigma) V_{I}\rangle \ d\sigma d\lambda.
\end{split}
\end{align}

\end{theorem}
\begin{proof}
Here we only need to prove \eqref{eq:limitGamma}. The other ones were proven in the context.

For $\Gamma_{k}$, we proved in \eqref{eq:half} that, for any Hermitian $T$, up to a correction of order $O(2^{-l})$, $D_{-2^{-l}}(T)$ is half of $T.$ Moreover $\text{Source}_{l}$ is of order $O(2^{2l})$ by \eqref{eq:sourcel}.

This implies the desired \eqref{eq:limitGamma}.

\end{proof}

\section{Strong subadditivity of quantum entropy}\label{sec:Mutual}
Suppose that $\rho_{ABC}$ is a positive-definite density matrix on the tensor product space $\mathcal{H}_A \otimes \mathcal{H}_B \otimes \mathcal{H}_C$.
We are interested in studying $S(\rho_{ABC} \mid \rho_{AB} \otimes \rho_C)-S(\rho_{BC} \mid \rho_{B} \otimes \rho_C)$, 
where $S(\rho_1|\rho_2)$, for any density matrices $\rho_1$ and $\rho_2$, is defined in \eqref{def:relative}, $\rho_{AB}$ and $\rho_B$ are defined as partial traces of $\rho_{ABC}:$
\begin{align*}
    \rho_{AB} &:= \text{Tr}_C \, \rho_{ABC}, \\
    \rho_B &:= \text{Tr}_{AB} \, \rho_{ABC},
\end{align*} and $\rho_{BC}$ and $\rho_{B}$ are defined similarly. 

It was proved by Lieb and Ruskai in \cite{LiebRuskai1973} that 
\begin{align}
    I(A:C|B)_\rho=S(\rho_{ABC} \mid \rho_{AB} \otimes \rho_C)- S(\rho_{BC} \mid \rho_{B} \otimes \rho_C)\geq 0.
\end{align} Recall the definition of quantum conditional mutual information $I(A:C|B)_\rho$ in \eqref{def:mutual}.

Before stating our result, we make some preparations to make Theorem \ref{MainTHM:convex} applicable.
We convert $\rho_{ABC}$ into $\frac{1}{\dim(\mathcal{H}_A)} I \otimes \rho_{BC}$ by applying a trace-preserving, completely positive map. While various methods could achieve this, such as integration using an appropriate Haar measure, we opt for a discrete approach in this paper, prioritizing a rapid solution over computational efficiency.

\begin{proposition}\label{prop:completelyPostive}
  
For any $N$, there exist finitely many unitary matrices
$$U_1, \dots, U_J $$ s.t. for any $ N \times N$ Hermitian matrix $A,$

$$\frac{1}{J} \sum_{k=1}^J U_k A U_k^* = \left( \frac{1}{N} \text{Tr} A \right) I_{N \times N}.$$

The map $A \rightarrow\frac{1}{J} \sum_{k=1}^J U_k A U_k^*$ is Trace-preserving, completely positive.
\end{proposition}
The proposition will be proved in subsection \ref{sub:completePo}.

We are ready to formulate our main result, which is Main Theorem \ref{MainTHM:Mutual} below.

Suppose that the subsystem $A$ is $N$-dimensional. For any density matrix $\rho_{ABC}$, we apply the mapping $\frac{1}{J} \sum_{k=1}^{J} U_k\ (\cdot)\ U_k^*$ on the subsystem $A$, to obtain
\begin{equation}
    \frac{1}{J} \sum_{k=1}^{J} U_k \rho_{ABC} U_k^* = \frac{1}{N} I_A \otimes \rho_{BC}.
\end{equation}

A useful observation is that
\begin{align}\label{eq:identityFJ}
    F_J:=S\left( \frac{1}{J} \sum_{k=1}^{J} U_k \rho_{ABC} U_k^* \ \bigg|\ \frac{1}{J} \sum_{k=1}^{J} U_k (\rho_{AB}\otimes \rho_{C}) U_k^* \right) &=  S\left( \rho_{BC} \ \bigg|\  \rho_B \otimes \rho_{C} \right) .
\end{align}

Next, we derive an expression for $F(J)-S(\rho_{ABC} \big| \rho_{AB}\otimes \rho_{C})$ by iteration.

To establish the iteration relation, we define, for any integer $K\in [2, J]$,
\begin{equation}\label{def:FK}
  F_K:=S\Big( \frac{1}{K} \sum_{k=1}^{K} U_k \rho_{ABC} U_k^* \ \Big|\ \frac{1}{K} \sum_{k=1}^{K} U_k (\rho_{AB} \otimes \rho_{C}) U_k^* \Big).
\end{equation}

To make the application of Theorem \ref{MainTHM:convex} transparent, we adopt the same notations by defining objects corresponding to those in \eqref{eq:FtA}:
\begin{align}\label{eq:choices}
\begin{split}
  A_1:=&U_K \rho_{ABC} U_K^*\\
  A_2:=&\frac{1}{K-1} \sum_{k=1}^{K-1} U_k \rho_{ABC} U_k^*\\
  B_1:=&U_K \rho_{AB}\otimes \rho_{C} U_K^*\\
  B_2:=&\frac{1}{K-1} \sum_{k=1}^{K-1} U_k \rho_{AB}\otimes \rho_C U_k^*.
\end{split}
\end{align}
Define a function
\begin{equation}
    F_K(t) := S\left( t A_1 + (1-t) A_2 \ \big|\ t B_1 + (1-t) B_2 \right).
\end{equation} 
Here we are interested in the case $t=\frac{1}{K}$, which makes
\begin{equation}
  F_{K}(\frac{1}{K})=F_{K}.
\end{equation}

Recall the definition of $\Gamma$ in \eqref{eq:GammaT}. To display its dependence on $K$, we use the notation $\Gamma_{K}$. And we define
\begin{equation}
\tilde{\Gamma}_{K}(t):=2(\log e) \big\langle V_{I},\ \Gamma_{K}V_{I}\big\rangle.
\end{equation}
By Theorem \ref{MainTHM:convex}, \eqref{eq:concavity} and that
for any integer $k\in [1,J]$,
\begin{equation*}
  S\big(  U_k \rho_{ABC} U_k^* \ \big|\  U_k (\rho_{AB} \otimes \rho_{C}) U_k^* \big)=S\big(   \rho_{ABC}  \ \big|\  \rho_{AB} \otimes \rho_{C}  \big),
\end{equation*}
we find the following iteration relation:
\begin{lemma} For any integer $K\in [2,J]$,
  \begin{align}\label{eq:FkIterate}
    F_{K}=\frac{1}{K} S(\rho_{ABC} \, | \, \rho_{AB} \otimes \rho_{C})+(1-\frac{1}{K})F_{K-1}-\frac{1}{K} \int_0^1 \int_{\frac{\lambda}{K}}^{\lambda} \tilde{\Gamma}_K(\sigma) \, d\sigma \, d\lambda.
  \end{align}
\end{lemma}

We are ready to state the main theorem:
For $2 \le K < J$, define
\begin{equation}
    \Delta_K := \frac{1}{K+1}  \int_{0}^{1} \int_{\frac{\lambda}{K}}^{\lambda} \tilde{\Gamma}_{K}(\sigma) \, d\sigma \, d\lambda\geq 0.
\end{equation} Recall the identity in \eqref{eq:identityFJ}. Iterate \eqref{eq:FkIterate} to find the following version of strong subadditivity of quantum entropy:
\begin{theorem}\label{MainTHM:Mutual}
  \begin{align}
      S( \rho_{BC} \ |\  \rho_B \otimes \rho_C )=S(\rho_{ABC} \, | \, \rho_{AB}\otimes \rho_C)-\frac{1}{J} \int_{0}^{1} \int_{\frac{\lambda}{J}}^{\lambda} \tilde{\Gamma}_J(\sigma) \, d\sigma \, d\lambda  - \sum_{K=2}^{J-1} \Delta_K.
  \end{align}
Equivalently, by the definition of mutual information in \eqref{def:mutual},
\begin{equation}\label{eq:mutual}
  I(A:C|B)_{\rho}=\frac{1}{J} \int_{0}^{1} \int_{\frac{\lambda}{J}}^{\lambda} \tilde{\Gamma}_J(\sigma) \, d\sigma \, d\lambda  + \sum_{K=2}^{J-1} \Delta_K\geq 0.
\end{equation}
\end{theorem}

\subsection{Proof of Proposition \ref{prop:completelyPostive}}\label{sub:completePo}
We begin by identifying a completely positive and trace-preserving map that transforms any Hermitian $N \times N$ matrix $A$ into a diagonal matrix.

Denote, by $\mathcal{T}$, the set of all subsets of $\{1,2,\cdots,N\}:$ 
\[
\mathcal{T} = \Big\{ T \ \Big|\ T \subset \{1, 2, 3, \dots, N\} \Big\}.
\]
Thus the total number of sets in $\mathcal{T}$ is $|\mathcal{T}| = 2^N$. Denote 
\begin{equation}\label{def:MathcalT}
\mathcal{T} = \{T_1, \dots, T_{2^N}\}
\end{equation}
 and for each fixed $T_{l}$, define a diagonal matrix
\begin{equation*}
K_{T_l} = \text{diag}[ k_{1}(T_l), \ k_{2}(T_{l}),\ k_{N}(T_{l})]
\end{equation*}
with $k_{m}(T_l)$, $m=1,\cdots,N$, defined as
\[
k_{m}(T_l) = 
\begin{cases} 
1 & \text{if } m \notin T_l ,\\ 
-1 & \text{if } m \in T_l.
\end{cases}
\]

We will use this observation: for any $j\in T_{l}$, $K_{T_{l}}A$ adds a minus-sign to the $j-$th column of $A$; and $A K_{T_{l}}$ adds a minus-sign to the $j-$th row. 

What is useful is that $\frac{1}{2^N} \sum_{T_l \in \mathcal{T}} K_{T_l} A K_{T_l}$ removes all the off-diagonal entries, but keeps the diagonal ones intact by the following reason: for any fixed $K_{T_j}$, apply it to the set $\big\{K_{T_{l}}\ |\ l=1,\cdots, 2^{N}\big\}$ will not change it:
\begin{equation*}
  K_{T_j} \big\{K_{T_{l}}\ |\ l=1,\cdots, 2^{N}\big\}=\big\{K_{T_{l}}\ |\ l=1,\cdots, 2^{N}\big\}.
\end{equation*} and thus
\begin{equation*}
  K_{T_j}\Big(\frac{1}{2^N} \sum_{T_l \in \mathcal{T}} K_{T_l} A K_{T_l}\Big) K_{T_j}=\frac{1}{2^N} \sum_{T_l \in \mathcal{T}} K_{T_l} A K_{T_l}.
\end{equation*}
Consequently, we get a diagonal matrix:
\begin{equation}\label{eq:diago}
\frac{1}{2^N} \sum_{T_l \in \mathcal{T}} K_{T_l} A K_{T_l} = \text{diag}[a_{11}, a_{22}, \dots, a_{nn}] =: \tilde{A}.
\end{equation}

This mapping is trace-preserving and completely positive because $K_{T_l}=K^{*}_{T_l}$ is unitary.

Next, we convert $\tilde{A}$ into $\frac{\text{Tr}(\tilde{A})}{N} Id$ by applying a trace-preserving completely positive map.

We consider the permutation group $G$ on the set $\{1, 2, \dots, N\}$, which has $N!$ elements. For each $P \in G$, we define a unitary matrix $D_P$ s.t. $D_P F$ permutes the $N$ rows of the matrix $F$ accordingly.
Then
\begin{equation}\label{eq:identity}
\frac{1}{N!} \sum_{P \in G} D_P \tilde{A} D^{*}_P = 
\frac{\text{Tr}(\tilde{A})}{N} I=\frac{Tr(A)}{N} I.
\end{equation}
Obviously $A\rightarrow\frac{1}{N!} \sum_{P \in G} D_P A D^{*}_P$ is a trace-preserving, completely positive map.

Combining \eqref{eq:diago} and \eqref{eq:identity}, we find that, for any $N\times N$ Hermitian matrix $A,$
\[
\frac{1}{2^N} \frac{1}{N!} \sum_{P \in G} \sum_{T_l \in \mathcal{T}} (D_P K_{T_l}) A (K_{T_l} D^{*}_P) = \frac{Tr(A)}{N} I.
\] Hence $A\rightarrow\frac{1}{2^N} \frac{1}{N!} \sum_{P \in G} \sum_{T_l \in \mathcal{T}} (D_P K_{T_l}) A (K_{T_l} D^{*}_P)$ is the desired map in Proposition \ref{prop:completelyPostive}.

\section*{Acknowledgement} The author wants to thank Joseph Renes of ETH Zurich for mentioning this problem.

\appendix

\section{ Proof of Lemma \ref{LM:base} }\label{sec:Base}
Recall that we are interested in the perturbation expansion of a Hermitian matrix
$
(A + \epsilon V)^q
$
where $q\in [0,1]$ is a scalar, $A > 0$, and $V$ is Hermitian, $\epsilon \in \mathbb{R}$ satisfies $|\epsilon| \ll 1$. It is well known that for some Hermitian matrices $V_q$ and $Q_q$,
\begin{equation}\label{eqn:Vp}
    (A + \epsilon V)^q = A^q + \epsilon V_q -\epsilon^2 Q_{q}+O(\epsilon^2).
\end{equation}

We find the desired $Q_q$ by using the standard theory, e.g. L\"owner-Heinz theorem, see Lemma \ref{LM:integral} below,
\begin{equation}\label{eq:Lowner}
(A+\epsilon V)^{q}
= \frac{\sin(\pi q)}{\pi } \int_{0}^{\infty} t^{q} \left( \frac{1}{t} - \frac{1}{t + A} \right)\, dt
\end{equation}
to find that
\begin{equation*}
    Q_q = \frac{\sin(\pi q)}{\pi} \int_{0}^{+\infty} t^{q} (t + A)^{-1} V (t + A)^{-1} V (t + A)^{-1} \, dt\geq 0.
\end{equation*}

For $V_q$, one can also obtain an expression from \eqref{eq:Lowner}. But we have to be careful with the cases $q=0, 1$ where the integrands are not integrable. For this reason, we derive an explicit expression for $V_q$ in the following way. 

For $q =1/2^{m},\ m \in \mathbb{N}$, we find $V_q$ recursively: because $$(A+\epsilon V)^{1/2}=A^{1/2 }+\epsilon V_{1/2}+O(\epsilon^2),$$ we have
\begin{equation}
  \big[A^{1/2}+\epsilon V_{1/2 }+O(\epsilon^2)\big]^2=\Big[(A+\epsilon V)^{1/2}\Big]^2=A+\epsilon V, 
\end{equation} thus
$V_{1/2}$ must satisfy the following Sylvester equation:
\begin{equation}\label{eqn:V12}
    A^{1/2} V_{1/2} + V_{1/2} A^{1/2} = V.
\end{equation}
Similarly, for any $m \ge 2$, $V_{1/2^m}$ is the solution to the equation
\begin{equation}\label{eqn:V2n}
    A^{1/2^m} V_{1/2^m} + V_{1/2^m} A^{1/2^m} = V_{1/2^{m-1}}.
\end{equation}

The result is the following:
\begin{lemma}\label{LM:Vp}

$V_{1/2^m}$, in \eqref{eqn:V12} and \eqref{eqn:V2n}, is well defined for any integer $m\geq 1$. 

There exists a unique Hermitian matrix $\widetilde{V}$ such that
\begin{equation}\label{eq:tildeV}
    \lim_{m \to \infty} 2^m V_{1/2^m} = \widetilde{V}.
\end{equation}
For any $0 \le q \le 1$,
\begin{equation}\label{eq:generalVq}
    V_q = \int_0^q A^s \widetilde{V} A^{q-s} ds.
\end{equation}

In particular, if $q = \frac{k}{2^m}$ for some $m, k \in \mathbb{N}$ and $0 \leq k \leq 2^m$,
\begin{equation}\label{eq:Vkn}
    V_{\frac{k}{2^m}} = \sum_{l=0}^{k-1} A^{\frac{l}{2^m}} V_{1/2^{m}} A^{\frac{k-1-l}{2^m}}=\sum_{l=0}^{2^j k-1} A^{\frac{l}{2^{m+j}}} V_{1/2^{m+j}} A^{\frac{2^j k-1-l}{2^{m+j}}}=\int_{0}^{\frac{k}{2^m}} A^s \widetilde{V} A^{\frac{k}{2^m}-s} ds.
\end{equation} where the second identity holds for any $j\in \mathbb{N}.$
\end{lemma}

\begin{proof}
We begin by proving that $V_{1/2^m}$ is well-defined.

We begin with setting up the problem.

By a unitary transformation, we can make $A$ diagonal. Thus it suffices to consider an easier case, where $A > 0$ is diagonal: for some $d_k>0,$ $k=1,\cdots,N,$ $$A = \text{diag}(d_1, \dots, d_N).$$

We suppose that the matrix $V$ is of the form $$V = [v_{k,l}]_{k,l=1}^N.$$
And we denote, for any nonnegative integer $j$
$$\delta_{j}=2^{-j}.$$

Solve for $V_{\delta_1}$ in \eqref{eqn:V12},
\[
V_{\delta_1} = \left[ \frac{1}{d_k^{\delta_1} + d_l^{\delta_1 }} v_{k,l} \right]_{k,l=1}^N
\]
Inductively, we solve \eqref{eqn:V2n} to find
\[
V_{\delta_m } = \left[ \frac{1}{\prod_{j=1}^m (d_k^{\delta_j } + d_l^{\delta_j })} v_{k,l} \right]_{k,l=1}^N
\]

Thus, we prove that $V_{\delta_m}$ is well defined.

Next, we prove \eqref{eq:tildeV}.

We observe that, as $j \to \infty$, $\delta_j\rightarrow 0$, and hence,
\[
d_k^{\delta_j } = 1 + \delta_j \ln d_k + O(\delta_j^2 ).
\]
Consequently, the following limit exists: for some constant $C_{k,l} > 0$
\[
\lim_{m \to \infty} \left[ \prod_{j=1}^m \frac{1}{d_k^{\delta_j } + d_l^{\delta_j }} v_{k,l}\right]_{k,l=1}^{N} 2^{-m} =\left[ C_{k,l}v_{k,l}\right]_{k,l=1}^{N}=:\tilde{V}.
\] Here $\tilde{V}$ is the desired one in \eqref{eq:tildeV}.

We are ready to prove \eqref{eq:Vkn}.

The idea is simple, $(A + \epsilon V)^{\frac{k}{2^m}}$ is a product of $k$ copies of $(A + \epsilon V)^{\frac{1}{2^m}}$: $$(A + \epsilon V)^{\frac{k}{2^m}}=\Big[(A + \epsilon V)^{\frac{1}{2^m}}\Big]^k=\Big[A^{\frac{1}{2^m}}+\epsilon V_{1/2^m}+O(\epsilon^2)\Big]^k=A^{\frac{k}{2^m}} +\epsilon\sum_{l=0}^{k-1} A^{\frac{l}{2^m}} V_{1/2^{m}} A^{\frac{k-1-l}{2^m}}+O(\epsilon^2),$$
which implies the first identity in \eqref{eq:Vkn}. For the second identity, we write $(A + \epsilon V)^{\frac{k}{2^m}}$ as a product of $k 2^{j}$ copies of $(A + \epsilon V)^{\frac{1}{2^{m+j}}}$. By arguing as finding the first identity, we obtain the second identity. For the third identity, we let $j\rightarrow \infty$, together with \eqref{eq:tildeV}, we turn the summation into an integral.

To prove \eqref{eq:generalVq}, we approximate $q\in (0,1)$ by $\frac{k}{2^m}$, then apply the third identity in \eqref{eq:Vkn}.

\end{proof}

\section{Integral formula}
We need the following integral formula, see e.g. \cite{carlen2010trace, bhatia97}.

\begin{lemma}\label{LM:integral}
For any positive definite matrix $A$, $p\in (0,1),$
\begin{equation}\label{eq:complex}
A^{p}
= \frac{\sin(p\pi)}{\pi } \int_{0}^{\infty} t^{p} \left( \frac{1}{t} - \frac{1}{t + A} \right)\, dt.
\end{equation}
\end{lemma}
\begin{proof}
  Here, since $A$ is positive definite, there exists a unitary matrix $U$ s.t.
\[ UAU^* = \text{diag} [a_1, \dots, a_n] \quad \text{with } a_k > 0. \]

To prove the desired formula \eqref{eq:complex}, we only need to show that, for any scalar $a > 0$,
\begin{align*}
    a^p &= \frac{\sin p\pi}{\pi} \int_{0}^{+\infty} t^p \left( \frac{1}{t} - \frac{1}{t+a} \right) \, dt \\
    &= \frac{\sin p\pi}{\pi} \int_{0}^{+\infty} t^{p-1} \frac{a}{t+a} \, dt \\
    &= a^p \frac{\sin p\pi}{\pi} D
\end{align*}
where, in the last step, we changed variable $t \to at$ and the constant $D\in \mathbb{R}$ is defined as 
\[ D := \int_{0}^{+\infty} t^{p-1} \frac{1}{1+t} \, dt. \]
Hence, what is left is to prove that
\begin{equation}\label{eq:claimD}
  D=\frac{\pi}{\sin p\pi}.
\end{equation}

Change variable $t = s^2$ to find:
\[ D = 2 \int_{0}^{+\infty} s^{2p-1} \frac{1}{1+s^2} \, ds .\]

Before applying complex analysis, we need to define $z^p$ properly. Here we consider the following region:
\begin{equation}
\Omega:=\Big\{   z = |z|e^{i\phi}, \Big|\  \phi \in \left( -\frac{\pi}{2}, \frac{3\pi}{2} \right)\Big\}
\end{equation} In this region we define $z^{q}$, $q\in (-1,1)$, as
\[  z^q = |z|^q e^{i q \phi}. \]

Hence for $s = -|s| = e^{i\pi} |s|$:
\begin{align*}
    s^{2p-1} &= |s|^{2p-1} \left( \cos(2p-1)\pi + i \sin(2p-1)\pi \right) \\
    &= -|s|^{2p-1} \left( \cos 2p\pi + i \sin 2p\pi \right).
\end{align*}
Thus,
\begin{align*}
  \int_{-\infty}^{0} s^{2p-1} \frac{1}{1+s^2} \, ds =& -(\cos 2p\pi + i \sin 2p\pi) \int_{-\infty}^{0} |s|^{2p-1} \frac{1}{1+s^2} \, ds \\
 =& -\frac{1}{2} (\cos 2p\pi + i \sin 2p\pi) D.
 \end{align*}

Consequently,
\begin{align}
    \int_{-\infty}^{+\infty} s^{2p-1} \frac{1}{1+s^2} \, ds &= \frac{D}{2} \left[ 1 - (\cos 2p\pi) - i \sin 2p\pi \right] \nonumber\\
    &= D \sin p\pi \left[ \sin p\pi - i \cos p\pi \right].\label{eq:frist}
\end{align}

On the other hand, in the upper complex half plane, $s = i = e^{i\frac{\pi}{2}}$ is the only singularity for the function $s^{2p-1} \frac{1}{1+s^2}$.

Hence, by the residue theorem:
\begin{align}
    \frac{1}{2\pi i} \int_{-\infty}^{+\infty} s^{2p-1} \frac{1}{1+s^2} \, ds &= \frac{1}{2i} \left( \cos(p\pi - \frac{\pi}{2}) + i \sin(p\pi - \frac{\pi}{2}) \right) \nonumber\\
    &= \frac{1}{2i} \left( \sin(p\pi) - i \cos(p\pi) \right). \label{eq:second}
\end{align}

\eqref{eq:frist} and \eqref{eq:second} imply that $D = \frac{\pi}{\sin p\pi}$ as desired \eqref{eq:claimD}.
\end{proof}


\section{Proof of the equivalence between \eqref{def:Q+-} and \eqref{eqn:Q+-}}\label{sec:LMDdelta}
We reformulate the problem into the next lemma.
\begin{lemma} 
Let $\Omega$ be a $N\times N$ Hermitian matrix, and $Y$ be a positive-definite $N\times N$ matrix. Then there exists a unique $N\times N$ Hermitian matrix $X$ such that
\begin{align}\label{eqn:DXOmega}
YX + XY = \Omega, 
\end{align} and it takes the form
\begin{equation}\label{eq:DXomega}
X = \int_{0}^{\infty} e^{-tY} \Omega e^{-tY} dt.
\end{equation}
In particular, if $\Omega > 0$ $(\geq 0)$ then $X > 0$ $(\geq 0)$.
\end{lemma}
\begin{proof}
Assuming the identity \eqref{eq:DXomega} holds and $\Omega\geq 0$, then since $e^{-tY} \Omega e^{-tY}\geq 0$ is for any $t$,
\[ X \geq 0. \]
By the same reason, if $\Omega > 0,\ (\geq 0)$, then $X > 0,\ (\geq 0)$.

What is left is to prove the identity \eqref{eq:DXomega}.

We start with considering an easy case, specifically, when $Y>0$ is a diagonal matrix, and hence, is of the form, for some $d_k>0$, $k=1,\cdots, N$,
\begin{equation}
  Y=\operatorname{diag}[d_1,d_2,\dots,d_N].
  \label{eqn:a}
\end{equation}

Denote the matrices $\Omega$ and $X$ by
\[
  \Omega=\bigl[w_{k,l}\bigr]_{k,l=1}^N,\ \text{and}\ 
  X=\bigl[x_{k,l}\bigr]_{k,l=1}^N.
\]
Then the equation takes the new form.
\[
  x_{k,l}\,(d_k+d_l)=w_{k,l}.
\]
Solve for $x_{k,l}$ to find
\[
  x_{k,l}
  =\frac{1}{d_k+d_l}\,w_{k,l}
  =\left(\int_0^\infty e^{-(d_k+d_l)t}\,dt\right) w_{k,l}.
\]
In the matrix form,
\begin{equation}\label{eqn:b}
  X=\int_0^\infty e^{-Yt}\,\Omega\,e^{-Yt}\,dt.
\end{equation}

Next, we consider the general case, i.e.\ $Y>0$, by converting it to the easier case \eqref{eqn:a} above.
Since $Y>0$, there exists a unitary matrix $U$ such that
\[
  \widetilde Y = UYU^*
\]
is diagonal, and hence is of the form \eqref{eqn:a}.

Apply $U$ and $U^*$ onto the appropriate places of equation \eqref{eqn:DXOmega} to find that
\[
  \widetilde Y\,\widetilde X+\widetilde X\,\widetilde Y=\widetilde\Omega,
\]
where $\widetilde X$ and $\widetilde\Omega$ are matrices defined as
\begin{equation}\label{eqn:XOmega}
  \widetilde X := UXU^*,
  \qquad
  \widetilde\Omega := U\Omega U^*.
\end{equation}

Since $\widetilde{D}$ is diagonal, the analysis for the easier case becomes applicable. By \eqref{eqn:b}, we find
\[
  \widetilde X = \int_0^\infty e^{-t\widetilde Y}\,\widetilde\Omega\,e^{-t\widetilde Y}\,dt.
\]
This and \eqref{eqn:XOmega}, together with applying $U^*$ and $U$ on the left and right, yield the desired identity
\[
  X = \int_0^\infty e^{-tY}\,\Omega\,e^{-tY}\,dt.
\]
\end{proof}

\section{Proof of \eqref{eq:half}}\label{sec:half}
Recall that $K_\delta \to I$ as $\delta \to 0$ and for some $C=C(\Psi,\Phi)>0$, $$\| K_\delta - I \| \leq C |\delta|.$$
To prove \eqref{eq:half} we only need to consider the cases where $|\delta|$ is small enough to make 
\begin{equation}\label{eq:smallDelta}
C |\delta|\leq \frac{1}{2}.
\end{equation} Otherwise, one can use the equivalent form \eqref{def:Q+-} to obtain a satisfactory estimate.

Intuitively, \eqref{eq:smallDelta} enables us to solve, perturbatively, the equation
\begin{equation}\label{eq:sys}
    K_\delta Y + Y K_\delta = F.
\end{equation}

To prove this rigorously, we rewrite \eqref{eq:sys} so that the standard Fixed Point Theorem applies, with the Banach space defined by the norm chosen in \eqref{def:Norm}. 

Define a new operator $L_{\delta}$, such that for any $N\times N$ Hermitian matrix $G$, $$K_\delta G+GK_{\delta} = 2G + L_\delta(G).$$ Then
\eqref{eq:sys} takes a new form
\[
Y = \frac{1}{2} F - \frac{1}{2}L_\delta(Y).
\]

The linear operator $\frac{1}{2} L_{\delta}$ is a contraction because $\| L_\delta \| \leq 1$ by \eqref{eq:smallDelta}. Apply the standard fixed point theorem to prove the desired \eqref{eq:half} by setting $Y=D_{\delta}(F)$:
\[
\| Y - \frac{1}{2} F \| \leq \frac{1}{2}\| L_\delta \| \| F \|\leq C \|F\| |\delta|.
\]



\end{document}